\documentclass[submission,copyright,creativecommons]{eptcs}
\providecommand{\event}{TLLA 2020} 
\usepackage{breakurl}             
\usepackage{underscore}           
\usepackage{graphicx}
\usepackage{adjustbox}
\usepackage{caption}
\usepackage{tikz}

\usepackage{savesym}
\usepackage{latexsym}
\usepackage{amssymb}
\usepackage{amsfonts}
\usepackage{amsmath}
\usepackage{amsthm}
\usepackage{cmll}

\usepackage{adjustbox}
\usepackage{txfonts}

\savesymbol{vartriangleleft}
\savesymbol{vartriangleright}

\usepackage[pdftex,noxy]{virginialake}

\usepackage{proofzilla}

\newcommand{\gA}{\mathcal A}
\newcommand{\gB}{\mathcal B}

\newcommand{\gG}{\mathcal G}
\newcommand{\gH}{\mathcal H}

\newcommand{\cgG}{{\gclr\cG}}
\newcommand{\cgH}{{\gclr\cH}}

\newcommand{\rbG}{\rbtrans{\gG}}
\newcommand{\crbG}{{\gclr\rbtrans{\gG}}}
\newcommand{\crb}[1]{{\gclr\rbtrans{#1}}}

\newcommand{\ccR}{{\xclr R}}
\newcommand{\cR}[1][]{{\xclr R}_{#1}}

\def\BV{\mathsf{BV}}
\def\X{\mathsf{X}}

\def\Xl{\mathsf{X^{\ell}}}

\def\Sfour{\mathsf{S4}}

\def\MLL{\mathsf{MLL}}

\def\MLLu{\mathsf{MLL_u}}
\def\MLLj{\mathsf{MLL_u^j}}

\def\MELL{\mathsf{MELL}}
\def\MELLj{\mathsf{MELL^j}}

\def\linearized#1{#1^{\ell}}

\def\cutr{\mathsf{cut}}
\def\mixr{\mathsf{mix}}

\def\Wrule{\mathsf{W}}
\def\Crule{\mathsf{C}}

\def\set#1{\{#1\}}

\def\tuple#1{\langle#1\rangle}

\newcommand{\RB}{{RB}}

\def\K{\mathsf K}

\newcommand{\MLLplus}[1]{{\MLL\mbox{-}{#1}}}

\newcommand{\MLLK}{\MLLplus\K}

\def\context#1{\Gamma\{#1\}}

\def\axrule{\mathsf {ax}}
\def\crule{\mathsf c}
\def\wrule{\mathsf w}

\def\krule{\mathsf k}

\def\RGB{RGB}

\def\rbtrans#1{{\partial({#1})}}

\def\atoms{\mathcal A}

\newvertex{vdotsnode}{\vdots}{}
\newvertex{dotsnode}{\dots}{}
\newvertex{bull}{\bullet}{}

\def\cons#1{\{#1\}}
\def\conhole      {\cons{\enspace}}%

\newemptyvertex{for}{}
\newvertex{PN}{\pi}{draw,circle}
\newvertex{bul}{\bullet}{}
\newvertex{circ}{\circ}{}

\newcommand{\contextbox}[2]{
	\tikz[overlay,remember picture]
	\draw[thick,rounded corners=2, densely dotted](#1\vertexcode.south west) 	rectangle (#2\vertexcode.north east);
}

\newcommand{\contextnodes}[1]{
	\foreach \aaa in {#1} {\contextbox{\aaa}{\aaa}}
}

\newtheorem{theorem}{Theorem}
\newtheorem{lemma}[theorem]{Lemma}
\newtheorem{proposition}[theorem]{Proposition}

\theoremstyle{definition}
\newtheorem{definition}[theorem]{Definition}
\newtheorem{definitionn}[theorem]{Definition}

\newtheorem{remark}[theorem]{Remark}

\newcommand{\form}[1]{\mathsf{fm}(#1)}
\newcommand{\cform}[1]{{\gclr\mathsf{fm}(#1)}}

\newbox\auxbox
\setbox\auxbox=\hbox{$\mathord\frown$}
\def\auxsym{\rlap{\copy\auxbox}\raise.3ex\hbox{\copy\auxbox}}
\newcommand{\auxedge}[1][]{\mkern1mu\mathord{\stackrel{#1}{{\rclr\auxsym}}}\mkern1mu}

\newcommand{\nedge}[1][]{\mkern1mu\mathord{\stackrel{#1}{{\smile}}}\mkern1mu}
\newcommand{\uedge}[1][]{\mkern1mu\mathord{\stackrel{#1}{{\rclr\frown}}}\mkern1mu}
\newcommand{\dedge}[1][]{\mkern1mu\mathord{\stackrel{#1}{{\mclr\rightsquigarrow}}}\mkern1mu}
\newcommand{\ndedge}[1][]{\mkern1mu\mathord{\stackrel{#1}{\not\rightsquigarrow}}\mkern1mu}

\newcommand{\bedge}[1][]{\mkern1mu\mathord{\stackrel{#1}{{\mclr\leftsquigarrow}}}\mkern1mu}

\newcommand{\nuedge}[1][]{\mkern1mu\mathord{\stackrel{#1}{\not\frown}}\mkern1mu}

\newcommand{\quand}{\quad\mbox{and}\quad}

\def\gclr{\color{linkcolor}}
\def\mclr{\color{modcolor}}
\def\rclr{\color{cographcolor}}
\def\jclr{\color{linkcolor}} 
\def\fclr{\color{skewcolor}}
\def\xclr{\color{orange}}

\def\widecneg#1{\overline{#1}}

\defedgetype{or}{dotted,black}{}

\def\derrule{\mathsf{der}_\wn}
\def\digrule{\mathsf{dig}_\wn }
\def\jdigrule{\mathsf{dig}_\ljump}

\def\dderrule{\deep\derrule}
\def\ddigrule{\deep\digrule}
\def\djdigrule{{\jclr \deep{\mathsf{dig}}_\ljump}}

\def\ewrule{\wrule_\wn}
\def\ecrule{\crule_\wn}

\def\prule{\mathsf \oc p}
\def\krule{\mathsf w\oc p}

\def\onerule{\mathsf 1}
\def\botrule{\mathsf \bot}

\def\dbotrule{{\jclr \deep \lbot}}
\def\dewrule{{\jclr \deep \ewrule}}
\def\decrule{\deep\crule_\wn}
\def\djcrule{{\jclr\deep\crule_\ljump}}

\def\jaxrule{\axrule_{{\jclr j}}}
\def\jonerule{\onerule_{{\jclr j}}}
\def\jbotrule{{\jclr \lbot^\mathsf j}}
\def\jewrule{{\jclr \mathsf{w^j}}}

\def\djbotrule{\deep\jbotrule}

\def\ljump{{\jclr \circ}}

\def\N{\mathbb N}

\newcommand{\relwebof}[1]{{\llbracket#1\rrbracket}}
\newcommand{\singlevertex}[1][]{\scriptstyle{#1}}

\def\lseq{\mathbin\origvartriangleleft}

\newcommand{\ssbot}{{\scriptstyle \bot}}
\newcommand{\ssone}{{\scriptstyle \one}}
\newcommand{\ssjump}{{\scriptstyle  \ljump}}

\newcommand{\sswn}{{\scriptstyle \wn}}
\newcommand{\ssoc}{{\scriptstyle \oc}}

\def\cL{\mathcal L}
\def\cA{\mathcal A}

\newvertex{ub}{\ssbot}{}
\newvertex{uj}{\ssjump}{}
\newvertex{uo}{\ssone}{}

\newcommand{\lab}[1]{l(#1)}
\newcommand{\mvertices}[1][]{V^{\ssoc\sswn}_{#1}}
\newcommand{\vertices}[1][]{V_{#1}}
\newcommand{\avertices}[1][]{V^\bullet_{#1}}
\newcommand{\bvertices}[1][]{V^\ssoc_{#1}}
\newcommand{\dvertices}[1][]{V^\sswn_{#1}}
\newcommand{\onevertices}[1][]{V^\ssone_{#1}}
\newcommand{\jumpvertices}[1][]{V^\ssjump_{#1}}

\newcommand{\nvertices}[1][]{V^\ast_{#1}}

\newcommand{\sizeof}[1]{\left|#1\right|}

\newcommand{\provesym}{\tikz[baseline=-.65ex]{\draw[very thick] (0,0)--(.5,0);\draw[thin] (0,-.1)--(0,.1);}}
\newcommand{\provevia}[1]{\mathrel{\stackrel{#1}{\provesym}}}

\newcommand{\ccf}{\fclr f}

\newcommand{\cf}{{\fclr f}}
\newcommand{\cg}{{\fclr g}}

\def\deep#1{#1^\downarrow}
\newcommand{\eqrule}{\equiv}
\newcommand{\deqrule}{\deep\equiv}

\def\one{1}

\def\Xj{\X^j}

\title{Exponentially Handsome Proof Nets and Their Normalization}

\author{Matteo Acclavio 
\institute{Department of Computer Science \\
University of Luxembourg}
\email{matteoacclavio.com}
}

\def\titlerunning{Exponentially Handsome Proof Nets and Their Normalization}
\def\authorrunning{Matteo Acclavio}

\usepackage{cleveref}
\usepackage{xspace}

\def\cgG{\gclr \gG}
\def\cgH{\gclr \gH}

\begin{document}

\maketitle

\begin{abstract}
Handsome proof nets were  introduced by  Retor\'e as a syntax for multiplicative linear logic.
These proof nets are defined by means of cographs (graphs representing formulas) equipped with a vertices partition satisfying simple topological conditions.
In this paper we extend this syntax to multiplicative linear logic with units and exponentials.
For this purpose we develop a new sound and complete sequent system for the logic,
enforcing a stronger notion of proof equivalence with respect to the one usually considered in the literature.
We then define combinatorial proofs, a graphical proof system able to capture syntactically the proof equivalence, for the cut-free fragment of the calculus.
We conclude the paper by defining the exponentially handsome proof nets as combinatorial proofs with cuts and defining an internal normalization procedure for this syntax.
\end{abstract}


\section{Introduction}

One of the novelties introduced by linear logic \cite{gir:ll} was the syntax of proof nets.
Proof nets are a graphical syntax \cite{lafont:89, lafont:95} for proofs
able to capture the \emph{proof equivalence} in the multiplicative fragment of linear logic (denoted $\MLL$):
proof nets are canonical representative of equivalent proofs modulo independent rules permutations.
In addition, proof nets are a sound and complete proof system in the sense of \cite{cook:reckhow:79} for $\MLL$, since it is possible to check if a graph represents a correct derivation in polynomial time with respect to the size of the graph. This test can be conducted by means of a topological criterion, often refereed to as \emph{correctness criterion} \cite{gir:ll,danos:regnier:89, guerrini:99, retore:03}.

Several  extensions of proof nets have been proposed to cover multiplicative linear logic with units ($\MLLu$),
but none of them can be considered to be fully satisfactory. 
In presence of the units, the correctness criterion requires to add additional edges, called \emph{jumps}, to a proof net in order to connect the gates of the unit $\lbot$ to an axiom or an unit $\lone$~\cite{hughes:simple-mult}.

The quest for a satisfactory syntax for $\MLLu$-proofs has come to an end after the publication of  \cite{heijltjes:houston:14} where is shown that it is not possible to have at the same time a syntax capturing the whole $\MLLu$ proof equivalence and a polynomial correctness criterion, unless $\mathsf{P}= \mathsf{NP}$. 
This result depends on the presence of the jumps:
on one hand they are needed in order to check in polynomial time if the proof net is correct, 
but on the other hand they enforce a coarse notion of proof equivalence which requires to ``rewire'' the jumps to capture the full proof equivalence.

A similar problem occurs in the multiplicative exponential linear logic ($\MELL$) due to the presence the weakening rule\footnote{Indeed, the decidability of $\MELL$ is still an open question and depends on the presence of the weakening rule~\cite{kopylov:decidability,str:decision}.} \cite{tortora:additives1,tortora:additives2}. 
Moreover, the presence of the promotion rule in this fragment poses an additional difficulty since
this rule is context-sensitive. 
However this latter problem is easily addressed by including in the proof net syntax the so called \emph{boxes} whose scope is to delimit portions of the graph~\cite{phd:regnier,phd:laurent,mazza:05,acc:proofd,acc:phd},
as shown in \Cref{fig:PNexample}.

\begin{figure}[h!]
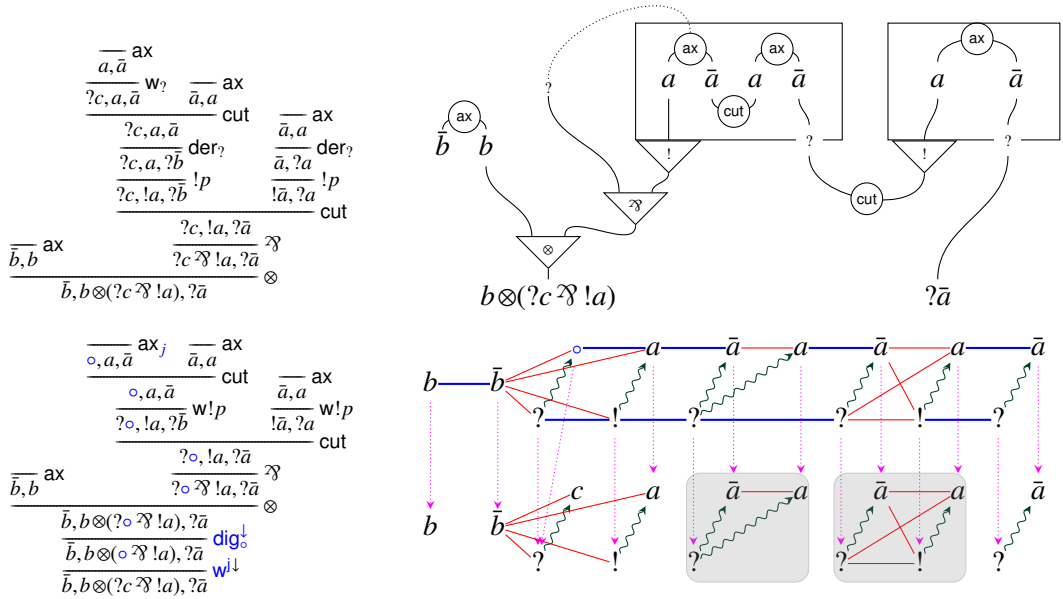

\centering
\begin{tabular}{cc}
	{\scriptsize
		$
	\vlderivation{
	\vliin{}{\ltens}{\cneg b, b\ltens(\wn c \lpar \oc a), \wn\cneg a}
	{\vlin{}{\axrule}{\cneg b, b}{\vlhy{}}}
	{
	\vlin{}{\lpar}{\wn c \lpar \oc a,\wn \cneg a}
	{\vliin{}{\cutr}{\wn c , \oc a, \wn \cneg a}
	{\vlin{}{\prule}{\wn c, \oc a,  \wn \cneg b}
		{\vlin{}{\derrule}{\wn c, a,\wn \cneg b}{{{\vliin{}{\cutr}{\wn c, a, \cneg a }{
				\vlin{}{\ewrule}{\wn c, a, \cneg a}{\vlin{}{\axrule}{a,\cneg a}{\vlhy{}}}
			}{\vlin{}{\axrule}{\cneg a, a}{\vlhy{}}}}}}}}
	{\vlin{}{{\prule}}{\oc \cneg a, \wn a}{
			{\vlin{}{\derrule}{\cneg a, \wn a}{\vlin{}{\axrule}{\cneg a, a}{\vlhy{}}}}}}}
	}}
	$}
&
	{
	\newvertex{wnna}{\wn \cneg a}{}
	\newvertex{myver}{b\ltens (\wn c\lpar \oc a)}{}
	$
	\begin{array}{c@{\quad}c@{\!\!\!\!}c@{\!\!\!\!}c@{\!\!}c@{\quad}c@{\quad}c@{\quad}cccc@{\qquad}ccc}
		&&&\boxYin1&& &&&&\boxYin2
		\\[5pt]
		&&\psanode[\wn]3&&\va1&\vna1&\va2&\vna2 &&&\va3&\vna3
		\\[10pt]
		\vnb1&\vb1&&&&&&&\boxYang1&&&&\boxYang2
		\\[10pt]
		&&&\Gpar1
		\\
		&&\Gtens1
		\\
	 && \vmyver1 &&&&&&&& \vwnna1
	\end{array}
	\psaxioms{a1/na1/1,a2/na2,b1/nb1,a3/na3}
	\psBox{1}{-140}{-40}
	\psBox{2}{-130}{-60}
	\pscuts{na1/a2/10,box1aux-40/box2main.O/10}
	\pswires{a1/box1main.I,a3/box2main.I,na2/box1aux-40,na3/box2aux-60,Gpar1.O/Gtens1.R,b1/Gtens1.L}
	\pswires{box1main.O/Gpar1.R,node3/Gpar1.L}
	\pswires{Gtens1.O/myver1,box2aux-60/wnna1}
	\psjump{node3}{ax1}{}
	$}
	\\[1em]
	{\scriptsize
		$
		\vlderivation{
			\vlin{}{\deep\jewrule}{\cneg b, b\ltens(\wn c \lpar \oc a), \wn\cneg a}{
			\vlin{}{\djdigrule}{\cneg b, b\ltens(\ljump \lpar \oc a), \wn\cneg a}{
			\vliin{}{\ltens}{\cneg b, b\ltens(\wn \ljump \lpar \oc a), \wn\cneg a}
			{\vlin{}{\axrule}{\cneg b, b}{\vlhy{}}}
			{
				\vlin{}{\lpar}{\wn \ljump \lpar \oc a,\wn \cneg a}
				{\vliin{}{\cutr}{\wn\ljump , \oc a, \wn \cneg a}
					{\vlin{}{{\krule}}{\wn\ljump, \oc a,  \wn \cneg b}
						{{
								{
									\vliin{}{\cutr}{\ljump, a, \cneg a }{
										{\vlin{}{\jaxrule}{\ljump,a,\cneg a}{\vlhy{}}}
									}{\vlin{}{\axrule}{\cneg a, a}{\vlhy{}}}}}}}
					{\vlin{}{{\krule}}{\oc \cneg a, \wn a}{
							{\vlin{}{\axrule}{\cneg a, a}{\vlhy{}}}}}}
		}}}}
	$}
	&
		$
	\begin{array}{cccccccccccccccccccccccc}
		&& &&\vuj1 &&\va1 && \vna1 &&\va9 && \vna9 && \va8 && \vna8 \\
		\vb1 &&\vnb1 \\
		&&&\vwn4&&\voc1 && 
		\vwn1
		&&&&\vwn2 &&\voc2 &&\vwn3 \\
		\\
		&&  &&\vc5 &&\va5 && \vna5 &&\va3 && \vna3 && \va6 && \vna6 \\
		\vb7 &&\vnb7 \\
		&&&\vwn8&&\voc5 && 
		\vwn5 
		&&&&\vwn6 &&\voc6 &&\vwn7 \\
	\end{array}
	\Bedges{uj1/a1, a1/na1, a9/na9, a8/na8,b1/nb1, wn4/oc1,oc1/wn1,wn1/wn2,oc2/wn3}
	\Medges{wn4/uj1,oc1/a1,wn2/na9,oc2/a8,wn3/na8}
	\Sedges{uj1/wn8,wn4/wn8, oc1/oc5,a1/a5,na1/na5,na9/na3,wn2/wn6,a9/a3,na8/na6,a8/a6,oc2/oc6,wn3/wn7,b1/b7,nb1/nb7}
	\modedges{wn8/c5,oc5/a5,wn6/na3,oc6/a6,wn7/na6}
	\multiRedges{nb1}{a1,oc1,wn4,uj1}
	\multiRedges{nb7}{c5,a5,oc5,wn8}
	\Redges{na1/a9}
	\Redges{na5/a3}
	\multiRedges{na9,wn2}{a8,oc2}
	\multiRedges{wn6,na3}{oc6,a6}
	\Sedges{wn1/wn5}
	\modedges{wn1/na1, wn1/a9}
	\modedges{wn5/na5,wn5/a3}
	\cutshade{wn5}{a3}
	\cutshade{wn6}{a6}
	$
	\end{tabular}
\caption{
	\textbf{Upper row}:
	a derivation of the sequent $\cneg d, d\ltens(\wn c \lpar \oc a), \wn(\cneg a \ltens \cneg b), \wn b$
	in $\MELL$
	and a corresponding proof net (with jump drawn as a dotted edge).
	\textbf{Lower row}:
	a decomposed derivation  in $\MELLj$ of the same sequent
	and its corresponding exponentially handsome proof net.
	In this latter, the gray shadings represent $\cutr$-rules, and the vertex $\ljump$ represents the jump.
	\vspace{-5pt}
} \label{fig:PNexample} \label{fig:HPNexample}
\end{figure}


\emph{Handsome proof nets} are an alternative syntax for $\MLL$ proofs introduced 
by Retor\'e in~\cite{retore:96:tokyo,retore:03} 
using the results 
from his PhD thesis
\cite{retore:phd}.
In this syntax, 
the information contained in a proof is represented by specific graphs encoding formulas 
(called \emph{cographs}\footnote{A cograph is a graph containing no induced subgraph isomorph to the four-vertices path $\mathsf P_4$. In \cite{duffin:65} it is shown that a graph encodes a formula iff it is a cograph.}) 
together with a 
perfect matching which satisfies specific topological conditions.
For this reason, we this syntax is called \emph{\RB-cographs}. 

\emph{Combinatorial proofs} 
were introduced by Hughes~\cite{hughes:pws,hughes:invar,str:07:RTA}
to capture proof equivalence in classical logic~\cite{str:FSCD17,acc:str:18}.
They can be considered as an extension of handsome proof nets since they are defined as 
specific graph homomorphisms from a \RB-cograph to a cograph encoding a formula.
In particular, combinatorial proofs indirectly give us a decomposition result,
allowing to  separate the ``linear'' part of the proof, that is,  
the part containing the logical interactions between its components, 
from
the ``resource management'' part, that is, 
the part taking care of erasing and duplicating components.

\textbf{Contribution of the paper.}
In this paper we extend Retor\'e's handsome proof nets for $\MLL$
and we define a syntax able to represent proofs in $\MLLu$ and $\MELL$.
For this purpose, we define combinatorial proofs
for cut-free $\MELL$-derivations,
and then we show how to also encode derivation with cuts.

To achieve our goal, 
we first definite of a new sound and complete proof system for $\MELL$,
called $\MELLj$, and we prove cut-elimination for it.
This system contains the \emph{weak promotion} rule from light linear logic \cite{girard:98,lafont2004soft,lau:pol}
and the \emph{digging rule} instead of regular promotion.
Moreover, the system uniquely assigns each unit $\lbot$ and each weakening rule to an branch of a derivation thanks to ad-hoc rules.
This choice allows us to mimic jumps assignation in proof nets and reduce the complexity of proof equivalence.

We then extend the results in \cite{acc:str:CPK} in order to define combinatorial proofs for $\MELLj$.
This allows us to represent equivalent cut-free derivations in $\MELLj$.
We show that this syntax has a polynomial correctness criterion and is able to syntactically capture the proof equivalence, that is, equivalent $\MELLj$-derivations are represented by the same syntactic object.

We conclude by defining \emph{exponentially handsome proof nets}
as extensions of the combinatorial proofs syntax.
Exponentially handsome proof nets allow us to encode derivations containing cuts, hence to compose proofs.
Finally, we provide a cut-elimination procedure for this syntax by means of a terminating graph rewriting.

\textbf{Organization of the paper.}
In \Cref{sec:formula}
we discuss the notion of proof equivalence for the standard $\MELL$ sequent calculus.
Then we define a sound and complete proof system $\MELLj$ for $\MELL$,
where $\MELL$-proof equivalence is restricted.
We then prove a decomposition theorem for $\MELLj$ using 
deep inference rules~\cite{gug:SIS,gug:str:01,brunnler:tiu:01}. 
In \Cref{sec:relweb,sec:RGB,sec:skew} we define the three components needed to define combinatorial proofs for $\MELLj$.
In particular, in \Cref{sec:relweb} we recall relation webs~\cite{gug:SIS,bechet:etal:97}, which generalize cographs, and we show how they can be used to encode formulas with modalities.
In \Cref{sec:RGB} we extend the correctness criterion for Retor\'e's \RB-cographs
to relation webs with special matchings, called \RGB-cographs, 
which encode the linear part of a $\MELLj$ proof.
In \Cref{sec:skew} we define the \emph{$\MELL$-fibrations} taking care of encoding the resource management part of our proofs.
In \Cref{sec:CP} we define combinatorial proofs as $\MELL$-fibrations from  an \RGB-cograph to a relation web.
Finally, in \Cref{sec:handsome}, 
we define exponentially handsome proof nets as $\MELL$-combinatorial proofs of sequents containing additional formulas keeping track of the $\cutr$-rules,
and we provide a cut-elimination procedure for this syntax.

\section{Proof Systems for $\MELL$}\label{sec:formula}

In this section we recall the sequent systems for multiplicative exponential linear logic (with units) and its subsystems.
We then discuss in \Cref{subsec:equivalence} the notion of proof equivalence for these logics
and in \Cref{subsec:restr} we define a sound and complete proof system for $\MELL$, 
enforcing a coarser notion of proof equivalence.
In \Cref{subsec:dec} we show that this new proof system admits a decomposition theorem,
which we exploit in \Cref{sec:CP} to define the syntax of combinatorial proofs for $\MELL$.

\begin{figure}[!t]
\begin{adjustbox}{max width=\textwidth}
	\def\myskip{\hskip1em}
		\begin{tabular}{c@{\myskip}c@{\myskip}c@{\myskip}c@{\myskip}c@{\myskip}c@{\myskip}c@{\myskip}c@{\myskip}c@{\myskip}|@{\myskip}c}
			$\vlinf{}{\axrule}{ a, \cneg a}{}$
			&
			$\vlinf{}{\lpar}{ \Gamma, A \lpar B}{ \Gamma, A , B}$
			&
			$\vliinf{}{\ltens}{ \Gamma, A \ltens B, \Delta}{ \Gamma, A}{ B , \Delta}$
			&
			$\vlinf{}{\botrule}{\Gamma, \bot}{ \Gamma}$
			&
			$\vlinf{}{\onerule}{\one}{}$    
			&
			$\vlinf{}{\prule}{ \oc A, \wn \Gamma}{ A, \wn \Gamma}$
			&
			$\vlinf{}{\derrule}{ \Gamma,  \wn A}{\Gamma,  A}$
			&
			$\vlinf{}{\ewrule}{ \Gamma, \wn A}{ \Gamma}$
			&
			$\vlinf{}{\ecrule}{ \Gamma, \wn A}{ \Gamma, \wn A, \wn A}$
			&
			$\vliinf{}{\cutr}{ \Gamma, \Delta}{ \Gamma, A}{\cneg A , \Delta}$
		\end{tabular}    
\end{adjustbox}
	\caption{Sequent calculus rules for $\MELL$ and the $\cutr$-rule}
	\label{fig:rules}
\end{figure}

We define formulas in negation normal form generated from 
a countable set of propositional variables $\atoms = \set{a, b, \dots}$
and set of constants $\set{\lone, \lbot, \ljump}$\footnote{The symbols $\lbot$ and $\lone$ are called units. The symbol $\ljump$ is a special symbol which we use as a ``placeholder'' for $\lbot$ and weakening rules.}
by the following grammar:
$$A,B::= a \mid \cneg a   \mid A\lpar  B \mid A\ltens B \mid \oc A \mid \wn A \mid \lbot \mid \lone\mid \ljump $$
A \emph{literal} is a formula of the shape $a$ or $\cneg a$ for an $a\in\atoms$.
A \emph{$\MELL$-formula} is  a formula containing no occurrences of $\ljump$.
Linear negation $\cneg \cdot$ is defined on $\MELL$ formulas through the De Morgan laws:
$\widecneg{\widecneg A}= A $,
$\widecneg{A \ltens B} = \widecneg A \lpar \widecneg B$,
$\widecneg {\oc A} = \wn \cneg A$,
$\cneg \lone = \lbot$.
A \emph{sequent} $\Gamma=A_1, \dots , A_n$ is a non-empty multiset of formulas.

In this paper we consider \emph{multiplicative linear logic} and its extensions with
\emph{units} and 
and \emph{exponentials}\footnote{In this paper, where not otherwise specified, we consider multiplicative exponential linear logic including units.} denoted respectively by $\MLL$, $\MLLu$ and $\MELL$ \cite{gir:ll}.
We use the same names to denote the cut-free sequent calculi for these logics defined using the sequent calculus rules in Figure \ref{fig:rules}.
$$
\MLL= \set{\axrule, \ltens, \lpar}
\qquad
\MLLu= \MLL\cup\set{\lbot, \lone}
\qquad
\MELL= \MLLu \cup \set{\prule, \ewrule,\ecrule,\derrule}
$$
We say that a formula $F$ is provable in $\X$ (denoted by ${\provevia \X F}$) if there is a derivation of $F$ in the system $\X$.

We recall the cut-elimination result for $\MELL$, which encompasses also $\MLL$ and $\MLLu$.
\begin{theorem}[\cite{gir:ll}]\label{thm:seq:cutelim}
	Let $F$ be a formula.
Then $F$ is provable in $\MELL$ iff it is provable in $\MELL\cup\set{\cutr}$.
\end{theorem}

\subsection{From Proof Nets to Handsome Proof Nets}\label{subsec:handoming}
\emph{Proof nets} are a graphical syntax for proofs in linear logic introduced by Girard in \cite{gir:ll}.
In this syntax, $\MLL$ proofs are represented by 
replacing each instance of a rule 
by a corresponding \emph{gate} 
whose 
inputs  are the active formulas of the rule 
and 
whose
outputs are the principal formulas, as shown in the upper row of \Cref{fig:PNgates}.
The graphs generated by these gates are called \emph{proof structures} and some of them do not have a corresponding proof in $\MLL$.
Therefore, a topological \emph{correctness criterion} needs to be defined to decide whether a proof structure is a proof net, i.e., is the translation of a $\MLL$ proof. 
Beside Girard's \emph{duable trip condition}, 
several alternative criteria have been proposed in the literature 
\cite{danos:regnier:89,guerrini:99}.

Using \emph{\RB-graphs}, 
which are graphs with two kind of edges 
({\rclr Red} or Regular, and {\color{blue}Blue} or Bold), 
Retor\'e defined in \cite{retore:96:tokyo}
the syntax of 
\emph{\RB-proof nets} 
where gates are represented as the \RB-graphs, as shown in the lower row of \Cref{fig:PNgates}.
In this syntax the correctness criterion can be formulated by requiring the absence of elementary (i.e., non self-intersecting) cycles made of alternating colours edges. 
From \RB-proof nets, 
he then discovered the syntax of 
\emph{handsome proof nets} (or \emph{\RB-cographs}) 
by using the transformation in \Cref{eq:handsome} below, 
allowing to remove from an \RB-graph
all the nodes which are not labelled by literals~\cite{retore:96:tokyo,ret:99,retore:03}.
This syntax is formally presented in \Cref{sec:RGB}.
Aim of this paper is to further develop \RB-cographs by adding an encoding of the linear logic modalities and units.
\begin{equation}\label{eq:handsome}
	\begin{array}{ccccc}
		\vbul1	& 		&	&		&\vbul2 \\
		\vdots	&\vcirc1&	&\vcirc2&\vdots \\
		\vbul3	&		&	&		&\vbul4
	\end{array}
	\Bedges{circ1/circ2}
	\Redges{bul1/circ1,bul3/circ1,bul2/circ2,bul4/circ2}
	\quad\rightsquigarrow\quad
	\begin{array}{ccccc}
		\vbul1	& 		&	&		&\vbul2 \\
		\vdots	&		&	&		&\vdots \\
		\vbul3	&		&	&		&\vbul4
	\end{array}
	\multiRedges{bul1,bul3}{bul2,bul4}
	\qquad
	\mbox{where the labels of the nodes $\circ$ are not literals}
\end{equation}
The encodings for modalities we employ in this paper uses the same kind of directed edges from 
the handsome proof nets for 
\emph{pomset logic}~\cite{retore:96:lambek,bechet:etal:97,retore:98,retore:99,retore:03,retore:21}
(we here represent these edges by green squiggly arrows instead of red arrows).
The upper row of the rightmost column of \Cref{fig:PNgates} shows
the usual encoding of the promotion rule of $\MELL$ by means of \emph{boxes} isolating a portion of the proof net.
In the lower row of the same column 
we provide the intuitive representation 
of how the same box should be represented 
in terms of the \RB-proof nets. 
This intuition might help the reader familiar with proof nets; however, it is not further developed in the paper, as our approach focuses on generalizing \RB-cographs.

\begin{figure}[t]
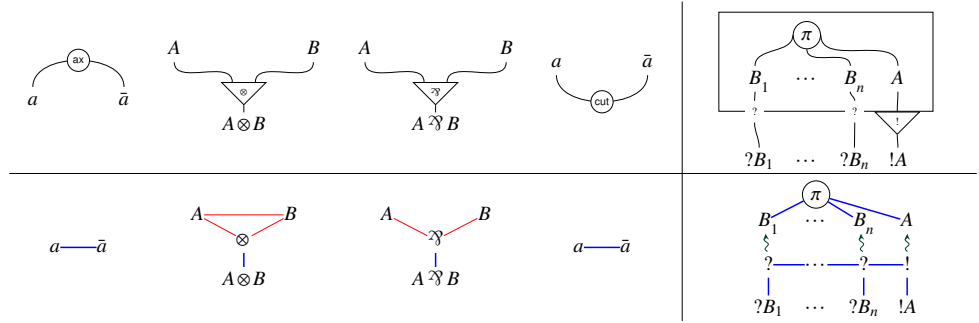

	$$
	\scalebox{.7}{$	
		\begin{array}{c@{\quad}c@{\quad}c@{\quad}c@{\quad}|@{\quad}c}
			\begin{array}{c@{\qquad}c@{\qquad}c}
				\\
				\va1 &&\vna1
			\end{array}
			\psaxioms{a1/na1}
		&
			\begin{array}{c@{\qquad}c@{\qquad}c}
				\vA1 &&\vB1\\[1em]
				&\Gtens1\\
				& \vfor1{A\ltens B}
			\end{array}
			\pswires{A1/Gtens1.L,B1/Gtens1.R,Gtens1.O/for1}
		&
			\begin{array}{c@{\qquad}c@{\qquad}c}
				\vA1 &&\vB1\\[1em]
				&\Gpar1\\
				& \vfor1{A\lpar B}
			\end{array}
			\pswires{A1/Gpar1.L,B1/Gpar1.R,Gpar1.O/for1}
		&	
			\begin{array}{c@{\qquad}c@{\qquad}c}
				\va1 &&\vna1
				\\
				\\
				\\
			\end{array}
			\pscuts{a1/na1}
		&
			\begin{array}{ccccccc}
				\boxYin1
				\\[-5pt]
						&		&\vPN1
				\\[10pt]
						&\vB1_1	&\cdots	&\vB2_n	&\vA1 
				\\[5pt]
						&		&		&		&		&\boxYang1&
				\\[10pt]
						&\vfor1{\wn B_1}&\cdots	&\vfor2{\wn B_n}	& \vfor3{\oc A}
			\end{array}
			\psBox{1}{-35}{-60,-145}
			\pswires{B1.O/box1aux-145.I,box1aux-145.O/for1.I}
			\pswires{B2.O/box1aux-60.I,box1aux-60.O/for2.I}
			\pswires{A1.O/box1main.I,box1main.O/for3.I}
			\pswires{PN1.180/B1,PN1.-90/B2,PN1.0/A1}
			\\\hline
			\begin{array}{ccc}
				\va1&				& 	\vna1\\
			\end{array}
			\Bedges{a1/na1}
			&
			\begin{array}{ccc}
				\vA1&				& 	\vB1\\
				&\vfor1{\ltens}				\\[8pt]
				&\vfor2{A\ltens B}
			\end{array}
			\Redges{A1/for1,B1/for1,A1/B1}	
			\Bedges{for1/for2}
			&
			\begin{array}{ccc}
				\vA1&				& 	\vB1\\
				&\vfor1{\lpar}				\\[8pt]
				&\vfor2{A\lpar B}
			\end{array}
			\Redges{A1/for1,B1/for1}	
			\Bedges{for1/for2}
			&
			\begin{array}{ccc}
				\va1&				& 	\vna1\\
			\end{array}
			\Bedges{a1/na1}
			&
			\begin{array}{ccccccc}
				\\[-10pt]
							& \vPN1
				\\
				\vB1_1		&\cdots			&\vB2_n	&\vA1 
				\\[10pt]
				\vfor4{\wn}&\vfor7{\cdots}	&\vfor5{\wn}		&\vfor6{\oc}
				\\[10pt]
				\vfor1{\wn B_1}&\cdots	&\vfor2{\wn B_n}	& \vfor3{\oc A}
			\end{array}
			\Bedges{for4/for7,for7/for5,for5/for6}
			\Bedges{for1/for4,for2/for5,for3/for6}
			\Medges{for4/B1,for5/B2,for6/A1}
			\Bedges{PN1/B1,PN1/B2,PN1/A1}
	\end{array}
	$}
	$$
	\vspace{-10pt}
	\caption{
		\textbf{Upper row}: gates for $\MLL$ proof nets and the way the promotion rule is encoded.
		\textbf{Lower row}: gates for \RB-proof nets and the intuition on how the promotion rule should be encoded in this syntax.
	}
	\label{fig:PNgates}
\end{figure}

\subsection{Proof Equivalence}\label{subsec:equivalence}

\def\npeq{\simeq}
\def\peq{\npeq_\mathsf J}

\begin{figure}[!t]
	\scriptsize
	\def\myskip{\hskip3em}
	\hbox to\textwidth{\hfil 
		\begin{tabular}{c@{\myskip}c@{\myskip}c}
			$
			\vlderivation{\vliin{}{\ltens}{\Gamma,\Delta, \Sigma, A\ltens B, C\ltens D}{\vlhy{\Delta, A}}{\vliin{}{\ltens}{\Delta, \Sigma ,B, C\ltens D}{\vlhy{\Delta,B,C}}{\vlhy{\Sigma, D}}}}
			\npeq
			\vlderivation{
				\vliin{}{\ltens}{\Gamma,\Delta, \Sigma, A\ltens B , C \ltens D}
				{\vliin{}{\ltens}{\Gamma, \Delta, A\ltens B , C}{\vlhy{\Delta, A}}{\vlhy{\Sigma,B,C}}}
				{\vlhy{\Sigma, D}}}
			$
			&
			$
			\vlderivation{\vlin{}{\tau}{\Gamma, A, B}{\vlin{}{\rho}{\Gamma, A, \Sigma}{\vlhy{\Gamma, \Delta, \Sigma}}}}
			\npeq
			\vlderivation{\vlin{}{\rho}{\Gamma, A, B}{\vlin{}{\tau}{\Gamma, \Delta, B}{\vlhy{\Gamma, \Delta, \Sigma}}}}
			$
			&
			$
			\vlderivation{\vliin{}{\ltens}{\Gamma,\Delta, A, B\ltens C}{\vlin{}{\rho}{\Gamma, A, B}{\vlhy{\Gamma, \Delta, B}}}{\vlhy{\Delta, C}}}
			\npeq
			\vlderivation{\vlin{}{\rho}{\Gamma,\Delta, A, B\ltens C}{\vliin{}{\ltens}{\Gamma, \Delta , \Delta, B\ltens C}{{\vlhy{\Gamma, \Delta, B}}}{\vlhy{\Delta, C}}}}
			$
		\end{tabular}
		\hfil}
	
	\caption{Independent rules permutations defined for all 
		$\rho, \tau \in \set{\lpar, \botrule, \ewrule, \ecrule, \derrule}$.
	}
	\label{fig:proofEq}
\end{figure}

\begin{figure}[!t]
	\scriptsize
\def\myskip{\hskip3em}
\hbox to\textwidth{\hfil 
	\begin{tabular}{c@{\myskip}c@{\myskip}c@{\myskip}c@{\myskip}c@{\myskip}c@{\myskip}c}
		$
		\vlderivation{\vlin{}{\ecrule}{\Gamma, \wn A}{\vlin{}{\ecrule}{\Gamma, \wn A_1 , \wn A}{\vlhy{\Gamma, \wn A_1, \wn A_2, \wn A_3}}}}
		\npeq
		\vlderivation{\vlin{}{\ecrule}{\Gamma, \wn A}{\vlin{}{\ecrule}{\Gamma, \wn A , \wn A_3}{\vlhy{\Gamma, \wn A_1, \wn A_2, \wn A_3}}}}
		$
		&
		$
		\vlderivation{\vlin{}{\ewrule}{\Gamma, \wn A,\wn A}{\vlin{}{\ecrule}{\Gamma, \wn A }{\vlhy{\Gamma, \wn A, \wn A}}}}
		\npeq
		\vlderivation{\vlhy{\Gamma,\wn A,\wn A}}
		$
		& 
		$
		\vlderivation{\vlin{}{\ecrule}{\Gamma, \wn A}{\vlin{}{\ewrule}{\Gamma, \wn A , \wn A}{\vlhy{\Gamma, \wn A}}}}
		\npeq
		\vlderivation{\vlhy{\Gamma,\wn A}}
		$
		\\[15pt]
		$
		\vlderivation{\vlin{}{\prule}{\wn \Gamma,\oc A, \wn B}{\vlin{}{\ewrule}{\wn \Gamma, A,\wn B}{\vlhy{\wn \Gamma, A}}}}
		\npeq
		\vlderivation{\vlin{}{\ewrule}{\wn \Gamma,\oc A, \wn B}{\vlin{}{\prule}{\wn \Gamma, \oc A}{\vlhy{\wn \Gamma, A}}}}
		$
		&&
		$
		\vlderivation{\vlin{}{\prule}{\wn \Gamma,\oc A, \wn B}{\vlin{}{\ecrule}{\wn \Gamma, A, \wn B}{\vlhy{\wn \Gamma, A, \wn B, \wn B}}}}
		\npeq
		\vlderivation{\vlin{}{\ecrule}{\wn \Gamma,\oc A, \wn B}{\vlin{}{\prule}{\wn \Gamma, \oc A, \wn B, \wn B}{\vlhy{\wn \Gamma, A, \wn B, \wn B}}}}
		$
	\end{tabular}
	\hfil}	
	\caption{Additional rules permutations.
		The first line is called the \emph{weakening-contraction comonad}}.
	\label{fig:proofEqNS}
\end{figure}

\def\sys{\mathcal S}

The notion of \emph{proof equivalence} in $\MELL$ is defined (e.g., in \cite{ret:DEA,danos:phd}) as the equivalence relation $\npeq$ over $\MELL$-derivations generated by \emph{independent rules permutations}, i.e., 
the permutations of rules which have disjoint sets of principal and active formulas shown in \Cref{fig:proofEq}, together with the ones in \Cref{fig:proofEqNS}.
It is worth noticing that these rules permutations are heavily used in the proof of cut-elimination theorems.

As shown in \cite{heijltjes:houston:14}, the instances of rule permutations in the last line of \Cref{fig:proofEq} involving $\botrule$-rule 
are responsible of 
the $\mathsf{P}$-$\mathsf{space}$ complexity bound of checking proof equivalence in $\MLLu$.
In fact, these rules permutations allows to move the $\lbot$ formula between different derivation branchings as shown in the following example,
where the leftmost and rightmost derivations, which are equivalent, are ``naturally'' represented by two different proof nets.

\begin{adjustbox}{max width=\textwidth}
	\newvertex{axr}{\axrule}{}
	\newvertex{botr}{\botrule}{}
	\newvertex{atenb}{\cneg a \ltens \cneg b}{}
$
\begin{array}{cccccc}
	\\
	\vlbot1	&\va1	&		&\Gtens1	&		&\vb1
	\\
	&		&\vatenb1	&		&	&
\end{array}
\pswires{Gtens1.O/atenb1.I}
\psaxioms{a1/Gtens1.L/1,b1/Gtens1.R}
\psjump{lbot1}{ax1}{}
\leftsquigarrow\quad
\vlderivation{
	\vliin{}{\ltens}{a, \cneg a\ltens \cneg b, b,\lbot}{\vlin{}{\vbotr1}{a, \cneg a, \lbot }{\vlin{}{\vaxr1}{a,\cneg a}{\vlhy{}}}}
	{\vlin{}{\axrule}{b,\cneg b}{\vlhy{}}}
}
\quad
\npeq
\quad
\vlderivation{\vlin{}{\botrule}{a, \cneg a\ltens \cneg b, b,\lbot}{\vliin{}{\ltens}{a,\cneg a\ltens \cneg b,a}
		{\vlin{}{\axrule}{a,\cneg a}{\vlhy{}}}
		{\vlin{}{\axrule}{a,\cneg a}{\vlhy{}}}
	}}
\quad
\npeq
\quad
\vlderivation{\vliin{}{\ltens}{a, \cneg a\ltens \cneg b, b,\lbot}
	{\vlin{}{\axrule}{a,\cneg a}{\vlhy{}}}
	{\vlin{}{\vbotr2}{b, \cneg b, \lbot }{\vlin{}{\vaxr2}{b,\cneg b}{\vlhy{}}}}
}
\bentRedges{axr1/botr1/100,axr2/botr2/20}
\quad\rightsquigarrow
\begin{array}{cc@{\quad}cc@{\quad}cc}
	\\
	\va1	&		&\Gtens1	&		&\vb1&\vlbot1
	\\
			&		&\vatenb1	&		&	&
\end{array}
\pswires{Gtens1.O/atenb1.I}
\psaxioms{a1/Gtens1.L,b1/Gtens1.R/1}
\psjump{lbot1}{ax1}{}
$
\end{adjustbox}\vspace{5pt}

As consequence of this complexity result, we cannot design a syntax $\sys $ for  $\MLLu$ satisfying the two following desiderata at the same time under the assumption $\mathsf{P}\neq \mathsf{NP}$:
\begin{itemize}
	
	\item correctedness in $\sys$ can be checked in polynomial time:
	we can check in polynomial time if an object expressed in the syntax $\mathcal S$ represents a correct proof in $\MLLu$;
	
	\item $\sys$ captures proof equivalence:
	if $\relwebof \pi$ and $\relwebof {\pi'}$ are the encodings in $\sys$ of two derivations $\pi$ and $\pi'$ in $\MELL$ such that $\pi\npeq\pi'$, then $\relwebof \pi=\relwebof{\pi'}$.

\end{itemize}

The same argument applies to $\MELL$-derivations in presence of the rule $\ewrule$.
The complexity of checking proof equivalence depends on the fact that each $\ewrule$ and $\botrule$ must be assigned to an instance of $\axrule$ or of $\onerule$ by permuting them upwards in a derivation. We refer to such assignation as  \emph{jump}.
Since $\npeq$ allows to re-assign jumps, the equivalence check has to test all possible jumps, whose number is exponential with respect to the number of $\botrule$ and $\ewrule$ in a derivation.

\subsection{Restricting Proof Equivalence in $\MELL$}\label{subsec:restr}

\begin{figure}[!t]
	\small
	\def\myskip{\hskip1em}
	\hbox to\textwidth{\hfil 
		\begin{tabular}{c@{\myskip}c@{\myskip}c@{\myskip}c@{\myskip}|@{\myskip}c@{\myskip}c@{\myskip}c}
			$\vlinf{}{\jaxrule^n }{ a, \cneg a, \ljump_1, \dots,  \ljump_n}{}$ 
			&
			$\vlinf{}{\jonerule^n }{\lone, \ljump_1, \dots,  \ljump_n}{}$
			&
			$\vlinf{}{\jbotrule}{\Gamma, \lbot}{\Gamma, \ljump}$
			&
			$\vlinf{}{\jewrule }{\Gamma, \wn A}{\Gamma, \ljump}$ 
			&
			$\vlinf{}{\krule}{ \oc A, \wn \Gamma}{ A, \Gamma}$
			&
			$\vlinf{}{\digrule}{ \wn A,\Gamma}{\wn \wn A,  \Gamma}$
			&
			$\vlinf{}{\jdigrule}{ \Gamma,   \ljump}{\Gamma, \wn \ljump}$
		\end{tabular}    
		\hfil}  
	\caption{On the left: the sequent rules fixing jump assignations. On the right: the weak promotion rule, the digging rule and the additional digging rule for $\ljump$}.
	\label{fig:jumprules}
\end{figure}

As consequence of \cite{heijltjes:houston:14}, 
we cannot aspire to design a proof system\footnote{In the sense of \cite{cook:reckhow:79}, that is, in which we can check if a syntactic object is correct in polynomial time.} 
which captures the proof equivalence $\npeq$ of $\MELL$.
To overcome this problem, 
in this subsection 
we define $\MELLj$, a sound and complete proof system for $\MELL$
enforcing a coarser proof equivalence, denoted by $\peq$, with respect to $\npeq$.
In fact, in $\MELLj$ each $\botrule$- and $\ewrule$-rule instance 
is uniquely associated to a specific $\axrule$- or $\onerule$-rule instance,
mimicking the way in which the corresponding nodes are linked by jumps in a proof net.

We define the following sequent systems using the rules in \Cref{fig:rules,fig:jumprules}
$$
\MLLj=\set{\jaxrule,\jonerule, \jbotrule,\lpar, \ltens}
\qquad
\MELLj=\set{\jaxrule,\jonerule, \jbotrule, \jewrule, \lpar, \ltens, \krule, \derrule, \digrule,\jdigrule, \ecrule}
$$
where $\jaxrule=\set{\jaxrule^n\mid n\in \N}$ and $\jonerule=\set{\jonerule^n\mid n\in \N}$.
The proof equivalence $\peq$ over $\MELLj$ derivations is defined as in \Cref{fig:proofEq} by considering $\rho$ and $\tau$ ranging over $\set{\jbotrule, \jewrule, \lpar, \derrule, \digrule,\jdigrule, \ecrule}$
plus the \emph{associativity of contraction}, that is, the rule permutation in the left-hand side of \Cref{fig:proofEqNS}.

In these systems the relation between one $\jbotrule$- or $\jewrule$-rule and one $\jaxrule$- or $\jonerule$-rule is 
syntactically encoded in the sequent system syntax, as jumps in proof nets are.
Each instance of $\jaxrule$ and $\jonerule$ introduces a bunch of jump place-holders denoted by $\ljump$.
Since each occurrence of a $\ljump$ is unique, each place-holder is further used by a single $\botrule$ or $\ewrule$ instance. 

\begin{remark}
 Another solution to uniquely associate $\botrule$ and $\ewrule$ to an axiom 
 would be to introduce an axiom rule with non-empty contexts, i.e., 
 a rule having as conclusion any sequent of the shape $\Gamma, a, \cneg a$,
 as done in sequent calculus for classical logic $\mathsf{G}$~\cite{troelstra:schwichtenberg:00}.
 Such axiom rule would keep track of whole information of the weakened formula,
 making the syntax of the structures described in \Cref{sec:RGB} heavier.
\end{remark}

Moreover, in $\MELLj$ we replace the \emph{promotion} rule $\prule$ with the \emph{weak promotion} rule  $\krule$ used in \emph{light linear logic}~\cite{girard:98} and in \emph{soft linear logic}~\cite{lafont2004soft,lau:pol} and the \emph{digging} rule.
This choice is motivated by the fact that weak promotion allows to group the $\oc$ introduced by a promotion with all the $\wn$ of its context formulas.
In particular, weak promotion is a context-free rule, that is, it can be applied independently form the shape of the premise context, 
while the regular promotion rule $\prule$ is a context-sensitive rule, in the sense that it can be applied only if the context is of the form $\wn \Gamma$.
The digging rules $\digrule$ and $\jdigrule$ are required to make the system sound and complete with respect to $\MELL$.

\begin{figure}[!t]
	\small
	\def\myskip{\hskip1em}
	\hbox to\textwidth{\hfil 
		\begin{tabular}{c@{\;}|@{\;}c@{\;}c}
			$
			\vlderivation{\vlin{}{\botrule}{\Gamma, \lbot}{\vlde{\pi}{}{\Gamma}{\vlin{}{\axrule}{a, \cneg a}{\vlhy{}}}}}
			\npeq
			\vlderivation{{\vlde{\pi,\lbot}{}{\Gamma,\lbot}{\vlin{}{\botrule}{a,\cneg a, \lbot}{\vlin{}{\axrule}{a, \cneg a}{\vlhy{}}}}}}
			\leftrightsquigarrow
			\vlderivation{{\vlde{\pi,\lbot}{}{\Gamma,\lbot}{{\vlin{}{\jbotrule}{a,\cneg a, \lbot}{\vlin{}{\jaxrule}{a, \cneg a, \ljump}{\vlhy{}}}}}}}
			$
			&
			$    \begin{array}{c}
				
				\vlinf{}{\prule}{\oc A, \wn \Gamma}{A, \wn \Gamma}
				\rightsquigarrow
				\vlderivation{
					\vliq{}{\digrule}{\oc A, \wn \Gamma}{
						\vlin{}{\krule}{\oc A, \wn \wn \Gamma}{\vlhy{A,\Gamma}}
				}}
				\\
				\vlinf{}{\krule}{\oc A, \wn \Gamma}{A, \Gamma}
				\rightsquigarrow
				\vlderivation{
					\vlin{}{\prule}{\oc A, \wn \Gamma}{
						\vliq{}{\derrule}{ A, \wn  \Gamma}{\vlhy{A,\Gamma}}
				}}
			\end{array}$
			
			&
			$
			\vlderivation{\vlin{}{\digrule}{\wn A, \Gamma}{\vlpr{\pi}{\MELLj}{\wn \wn A, \Gamma}}}
			\rightsquigarrow
			\vlderivation{
				\vliin{}{\cutr}{\wn A, \Gamma}
				{\vlpr{\pi'}{\MELL}{\wn \wn A, \Gamma}}
				{\vliq{}{2\times \prule}{\oc\oc\cneg A, \wn A}{{\vlin{}{\derrule}{\cneg A,\wn A}{\vlin{}{\axrule}{\cneg A, A}{\vlhy{}}}}}
			}}
			\underset{\cutr}{\rightsquigarrow^*}
			\vlderivation{\vlpr{\pi''}{\MELL}{\wn A, \Gamma}}
			$
		\end{tabular}    
		\hfil}  
	\caption{How to transform a $\botrule$ associate to a specific $\axrule$-rule  to a $\jbotrule$ associate to a $\jaxrule$-rule, and how to replace $\prule$ with a derivation using $\krule$ and $\digrule$ and vice versa. Note that elimination of the $\digrule$ introduces a $\cutr$, which can be eliminated relying on the cut-elimination result in $\MELL$.}
	\label{fig:transformations}
\end{figure}

\begin{theorem}\label{prop:jumps}
	If $F$ is a fomula, then
	$\provevia \X F $ iff $ \provevia \Xj F$.
\end{theorem}
\begin{proof}
	By rules permutations, we can move each occurrence $\rho$ of a $\ewrule$- or a $\botrule$-rule up in the derivation until it reaches the assigned occurrence $\sigma_\rho$ of an $\jaxrule$- or a $\jonerule$-rule. 
	Then we replace $\sigma_\rho$ with an occurrence of the same rule $\sigma_\rho$ with an additional $\ljump$ in the conclusion, and $\rho$ with an occurrence of $\jclr \rho^\mathsf j$
	applied to this fresh $\ljump$ (an example is shown in Figure \ref{fig:transformations}). 
	Moreover, every $\prule$ can be replaced by a $\krule$ followed by a finite number of $\digrule$ and vice versa by $\MELL$ cut-elimination theorem \cite{gir:ll} (see Figure \ref{fig:transformations}).
\end{proof}

Strictly speaking, the proof system $\MELLj$ does not satisfy the \emph{subformula property} because of the presence of the rules $\digrule$ and $\jdigrule$.
However, we can prove a cut-elimination result.
Observe that a weaker notion of the subformula property holds since all formulas that can appear in a derivation of an arbitrary sequent are going to be subformulas of a formula at the root or formulas of the form $\wn\cdots \wn A$ where $\wn A$ is a subformula of a formula at the root.
At the moment this paper is written, the decidability of $\MELL$ is an open question~\cite{str:decision}. 
Thus we do not focus on proof search for $\MELLj$.
\def\cutred{\rightsquigarrow_\cutr}
\def\jruler{\rho_{\jclr j}}
\begin{figure}[t]
	\scriptsize
	\centering
	$
	\begin{array}{cc}
		\vlderivation{\vliin{}{\cutr}
			{a,\cneg a, \ljump,\cdots, \ljump,\ljump',\cdots, \ljump'}
			{\vlin{}{\jaxrule}{a,\cneg a, \ljump, \cdots, \ljump}{\vlhy{}}}
			{\vlin{}{\jaxrule}{a,\cneg a, \ljump',\cdots, \ljump'}{\vlhy{}}}
		}
		\cutred
		\vlinf{}{\jaxrule}{a,\cneg a, \ljump, \cdots, \ljump,\ljump',\cdots, \ljump'}{}
		&
		\vlderivation{\vliin{}{\cutr}
			{\Xi, \ljump,\cdots, \ljump,\ljump',\cdots, \ljump'}
			{\vlin{}{\jbotrule}
				{\Xi, \ljump, \cdots, \ljump, \lbot}
				{\vlin{}{\jruler}{\Xi, \ljump,\cdots, \ljump, \ljump^{\cutr}}{\vlhy{}}}
			}
			{\vlin{}{\jonerule}{\lone, \ljump',\cdots, \ljump'}{\vlhy{}}}}
		\cutred
		\vlderivation{
			{\vlin{}{\jruler}{\Xi, \ljump, \cdots, \ljump,\ljump',\dots, \ljump'}{\vlhy{}}}}
		\\[1em]
		\\
		\vlderivation{\vliin{}{\cutr}{\Gamma, \Delta,\Sigma}{\vlin{}{\lpar}{\Gamma, A\lpar B}{\vlhy{\Gamma, A,B}}}{
				\vliin{}{\ltens}{\Delta, \Sigma, \cneg A\ltens \cneg B}{\vlhy{\Delta, \cneg A}}{\vlhy{\Sigma, \cneg B}}}}		
		\cutred
		\vlderivation{
			\vliin{}{\cutr}{\Gamma, \Delta, \Sigma}{
				\vliin{}{\cutr}{\Gamma, \Delta,B}{\vlhy{\Gamma, A,B}}{\vlhy{\Delta, \cneg A}}
			}
			{\vlhy{\Sigma, \cneg B}}
		}
		&
		\vlderivation{\vliin{}{\cutr}{\wn\Gamma, \wn \Delta, \oc B}{\vlin{}{\krule}{\wn \Gamma, \oc A}{\vlhy{\Gamma, A}}}{\vlin{}{\krule}{\wn  \cneg A,\wn \Delta,  \oc B}{\vlhy{\cneg A,\Delta, B}}}}
		\cutred
		\vlderivation{\vlin{}{\krule}{\wn \Gamma, \wn\Delta, \oc B}{\vliin{}{\cutr}{\Gamma, \Delta, B}{\vlhy{\Gamma, A}}{\vlhy{\cneg A,\Delta, B}}}}
		\\[1em]
		\\
		\vlderivation{\vliin{}{\cutr}{\wn\Gamma,  \Delta}{\vlin{}{\krule}{\wn \Gamma, \oc A}{\vlhy{\Gamma, A}}}{\vlin{}{\derrule}{\wn  \cneg A,\Delta}{\vlhy{\cneg A,\Delta}}}}
		\cutred
		\vlderivation{\vliq{}{\sizeof{\Gamma}\times \derrule}{\wn \Gamma,\Delta}{\vliin{}{\cutr}{\Gamma, \Delta}{\vlhy{\Gamma, A}}{\vlhy{\cneg A,\Delta}}}}
		&
		\vlderivation{\vliin{}{\cutr}{\wn\Gamma,  \Delta}{\vlin{}{\krule}{\wn \Gamma, \oc A}{\vlhy{\Gamma, A}}}{\vlin{}{\ecrule}{\wn  \cneg A,\Delta}{\vlhy{\wn\cneg A, \wn\cneg A,\Delta}}}}
		\cutred
		\vlderivation{\vliq{}{\sizeof{\Gamma}\times \ecrule}{\wn \Gamma, \Delta}{\vliin{}{\cutr}{\wn\Gamma,\wn\Gamma,  \Delta}{\vlin{}{\krule}{\wn \Gamma, \oc A}{\vlhy{\Gamma, A}}}{\vliin{}{\cutr}{\wn A, \wn\Gamma,  \Delta}{\vlin{}{\krule}{\wn \Gamma, \oc A}{\vlhy{\Gamma, A}}}{{\vlhy{\wn\cneg A, \wn\cneg A,\Delta}}}}}}
		\\
		\vlderivation{\vliin{}{\cutr}{\wn\Gamma,  \Xi}
			{\vlin{}{\krule}{\wn \Gamma, \oc A}{\vlhy{\Gamma, A}}}
			{\vlin{}{\jewrule}{\wn  \cneg A,\Xi}{
					{\vlin{}{\jruler}{\ljump^{\cutr},\Xi}{\vlhy{}}}}
		}}
		\cutred
		\vlderivation{\vliq{}{\sizeof \Gamma \times \jewrule}{\wn \Gamma, \Xi}{
				{\vlin{}{\jruler}{\ljump_1,\dots, \ljump_{\sizeof{\Gamma}}, \Xi}{\vlhy{}}}
			}
		}
		&
		\vlderivation{\vliin{}{\cutr}{\wn\Gamma,  \Delta}{\vlin{}{\krule}{\wn \Gamma, \oc A}{\vlhy{\Gamma, A}}}
			{\vliq{}{n\times \digrule}{\wn  \cneg A,\Delta}{
					\vliq{}{(n+1)\times \krule}{\wn^{n+1}\cneg A,\Delta}{
						{\vlhy{\cneg A,\Delta'}}}
		}}}
		\cutred
		\vlderivation{
			\vliq{}
			{n\sizeof{\Gamma}\times \digrule}
			{\wn \Gamma, \Delta}
			{\vliq{}
				{(n+1)\times \krule}
				{\wn^{n+1}\Gamma, \Delta}
				{\vliin{}
					{\cutr}
					{\Gamma, \Delta'}
					{\vlhy{\Gamma, A}}
					{\vlhy{\cneg A,\Delta'}
				}}
		}}
	\end{array}
	$
	\caption{The cut-elimination steps in $\MELLj$, where $\jruler\in\set{\jaxrule,\jonerule}$ is the unique rule introducing the $\ljump^\cutr$ which is the active premise of the $\jbotrule$ or $\jewrule$ involved in the cut-elimination step.}
	\label{fig:cutElim}
\end{figure}
\begin{figure}[!th]
	\small
	\centering
	$
	\begin{array}{c@{\cutred}c@{\qquad}c@{\cutred}c@{\qquad}c@{\cutred}c}
		\vlderivation{\vlin{}{\krule}{\wn \wn A, \wn \Gamma, \oc B}{\vlin{}{\derrule}{\wn A, \Gamma, B}{\vlhy{A,\Gamma, B}}}}
		&
		\vlderivation{\vlin{}{\derrule}{\wn \wn A, \wn \Gamma, \oc B}{\vlin{}{\krule}{\wn A, \wn \Gamma, \oc B}{\vlhy{A,\Gamma, B}}}}
		&
		\vlderivation{\vlin{}{\krule}{\wn \wn A, \wn \Gamma, \oc B}{\vlin{}{\digrule}{\wn A, \Gamma, B}{\vlhy{\wn\wn A,\Gamma, B}}}}
		&
		\vlderivation{\vlin{}{\digrule}{\wn \wn A, \wn \Gamma, \oc B}{\vlin{}{\krule}{\wn \wn \wn A, \wn \Gamma, \oc B}{\vlhy{\wn \wn A,\Gamma, B}}}}	
		\\
		\vlderivation{\vlin{}{\krule}{\wn X, \wn \Gamma, \oc B}{\vlin{}{\jewrule/\jbotrule}{X, \Gamma, B}{\vlhy{\ljump,\Gamma, B}}}}
		&
		\vlderivation{\vlin{}{\jewrule}{\wn X, \wn \Gamma, \oc B}{\vlin{}{\digrule}{\ljump, \wn \Gamma, \oc B}{\vlin{}{\krule}{\wn \ljump, \wn \Gamma, \oc B}{\vlhy{\ljump,\Gamma, B}}}}}
		&
		\vlderivation{\vlin{}{\krule}{\wn \wn A, \wn \Gamma, \oc B}{\vlin{}{\ecrule}{\wn A, \Gamma, B}{\vlhy{\wn A, \wn A,\Gamma, B}}}}
		&
		\vlderivation{\vlin{}{\ecrule}{\wn \wn A, \wn \Gamma, \oc B}{\vlin{}{\krule}{\wn \wn A, \wn \wn A, \wn \Gamma, \oc B}{\vlhy{\wn A, \wn A,\Gamma, B}}}}	
		\\
		\vlderivation{\vlin{}{\derrule}{\wn X, \Gamma}{\vlin{}{\jewrule/\jbotrule}{X, \Gamma}{\vlhy{\ljump,\Gamma}}}}
		&
		\vlderivation{\vlin{}{\jewrule}{\wn X, \Gamma}{{\vlhy{\ljump,\Gamma}}}}
		&
		\vlderivation{\vlin{}{\digrule}{\wn A, \Gamma}{\vlin{}{\jewrule}{\wn \wn A, \Gamma}{\vlhy{\ljump,\Gamma}}}}
		&
		\vlderivation{\vlin{}{\jewrule}{\wn A, \Gamma}{{\vlhy{\ljump,\Gamma}}}}
		\\
		\vlderivation{\vlin{}{\derrule}{\wn \wn A, \Gamma}{\vlin{}{\ecrule}{\wn A, \Gamma}{\vlhy{\wn A, \wn A,\Gamma}}}}
		&
		\vlderivation{\vlin{}{\ecrule}{\wn  \wn A,  \Gamma}{\vliq{}{2\times \derrule}{\wn\wn  A,\wn \wn  A,\Gamma}{\vlhy{\wn A, \wn A,\Gamma}}}}	
		&
		\vlderivation{\vlin{}{\digrule}{\wn  A,  \Gamma}{\vlin{}{\ecrule}{\wn \wn A, \Gamma}{\vlhy{\wn \wn A, \wn\wn A,\Gamma}}}}
		&
		\vlderivation{\vlin{}{\ecrule}{\wn A,  \Gamma}{\vliq{}{2\times\digrule}{\wn  A, \wn  A,  \Gamma}{\vlhy{\wn\wn A, \wn\wn A,\Gamma}}}}	
	\end{array}
	$
	
	$
	\begin{array}{c@{\cutred}c@{\qquad}c@{\cutred}c@{\qquad}c@{\cutred}c@{\qquad}c@{\cutred}c@{\qquad}c@{\cutred}c@{\qquad}c@{\cutred}c}
		\vlderivation{\vlin{}{\derrule}{\wn\wn A,\Gamma}{\vlin{}{\digrule}{\wn A,\Gamma}{\vlhy{\wn\wn A, \Gamma}}}}
		&
		\vlderivation{\wn\wn A,\Gamma}
		&
		\vlderivation{\vlin{}{\derrule}{\wn\ljump,\Gamma}{\vlin{}{\jdigrule}{\ljump,\Gamma}{\vlhy{\wn\ljump, \Gamma}}}}
		&
		\vlderivation{\wn\ljump,\Gamma}
		&
		\vlderivation{\vlin{}{\digrule}{\wn A,\Gamma}{\vlin{}{\derrule}{\wn \wn A,\Gamma}{\vlhy{\wn A, \Gamma}}}}
		&
		\vlderivation{\wn A,\Gamma}
		&
		\vlderivation{\vlin{}{\jdigrule}{\ljump,\Gamma}{\vlin{}{\derrule}{\wn\ljump,\Gamma}{\vlhy{ \ljump, \Gamma}}}}
		&
		\vlderivation{\ljump,\Gamma}
	\end{array}
	$
	\caption{Commutative cut-elimination steps}
	\label{fig:commCut}
\end{figure}

\begin{theorem}\label{thm:seq:jumpcut}
	Let $F$ be a formula.
	Then  $F$ is provable in $\MELLj$ iff $F$ is provable in $\MELLj\cup\set\cutr$.
\end{theorem}
\begin{proof}
	Assuming the equivalences $\peq$, 
	we define the rewritings in \Cref{fig:cutElim}, 
	which decrease the size of the cut-formula, 
	and the ones in \Cref{fig:commCut},
	which 
	permute the rule 	$\krule $ 								above $\derrule$, $\digrule$, $\jdigrule$, $\jewrule$, $\jbotrule$ and  $\ecrule$,
	permute the rules 	$\derrule$, $\digrule$ and $\jdigrule$ 	above $\jewrule$, $\jbotrule$ and  $\ecrule$,
	and	
	permute the rules $\jewrule$ and $\jbotrule$ 			above  $\ecrule$.
	
	Note that the cut-elimination steps involving $\ljump$ 
	are non-local rewritings which reassign or introduce new $\ljump$ in $\jaxrule$ or $\jonerule$ rules.
	Moreover, after the rules in \Cref{fig:commCut}, 
	the step involving $\digrule$ has to involve at least two different $\krule$.
	
	To prove cut-elimination we define some rewriting steps behaving similarly to the ones used in \cite{accattoli:linear}, that is, where rewriting deals with multiple boxes at a time.
\end{proof}

\subsection{Decomposing $\MELL$ Proofs}\label{subsec:dec}

In order to prove the decomposition result for our system $\MELLj$, 
we introduce \emph{deep inference} rules~\cite{gug:SIS,gug:str:01,brunnler:tiu:01}.
These rules can be applied at any depth of the sequent, i.e., to any subformula occurring in it,
allowing us to push to the bottom of the derivation the $\ewrule$, $\ecrule$, $\digrule$ and $\derrule$ inferences.

\begin{figure}[!t]
	\def\myskip{\hskip1em}
	\hbox to\textwidth{\hfil 
		\begin{tabular}{c@{\myskip}c@{\myskip}c@{\myskip}|@{\myskip}c@{\myskip}c@{\myskip}|@{\myskip}c@{\myskip}c}
			$     \vlinf{}{\dderrule}{\context{\wn A}}{\context{ A}}$
			&
			$     \vlinf{}{\ddigrule}{\context{\wn A}}{\context{\wn \wn A}}$
			&
			$     \vlinf{}{\djdigrule}{\context{\ljump}}{\context{\wn \ljump}}$
			&
			$     \vlinf{}{\dbotrule}{\context{\lbot}}{\context{ \ljump}}$
			&      
			$     \vlinf{}{\dewrule}{\context{ \wn A}}{\context{\ljump}}$
			&
			$     \vlinf{}{\decrule}{\context{\wn A}}{\context{\wn A\lpar \wn A}} $
		\end{tabular}
		\hfil}
	\caption{Deep inference rules for 
		dereliction and digging, for $\botrule$ and $\wn$-weakening, and  for $\wn$-contraction.
	}
	\label{fig:DI}
\end{figure}

We denote by  $\context{~}$  a \emph{context}, which is a sequent or a formula with an ``hole'' in place of a formula,
and we define the following sets of rules, composed of sequent rules from \Cref{fig:rules} and the deep inference rules from \Cref{fig:DI}.
\begin{equation}\label{eq:deepsys}
	\begin{array}{r@{\;=\;}l@{\qquad}r@{\;=\;}l@{\qquad}r@{\;=\;}l}
	\linearized\MLL 	& 	\MLL 
	&
	\linearized\MLLu 	& 	\set{ \jaxrule,\jonerule, \lpar, \ltens}
	&
	\linearized\MELL 	&	\set{\jaxrule, \jonerule,\lpar, \ltens, \krule}
	\\
	\deep\MLL 			& 	\emptyset
	&
	\deep\MLLu 			& 	\set{\dbotrule}
	&
	\deep\MELL 			&	\set{\dderrule, \ddigrule,\djdigrule, \dbotrule,\dewrule, \decrule}
	\\
	\end{array}
\end{equation}

\begin{theorem}\label{thm:MELL-decompose}
  Let  $F$ be a formula.
  Then $  \provevia{\MELL} F $ iff there is a formula $F'$ such that $ \provevia {\linearized \MELL} F' \provevia{\deep\MELL}F $.
  More precisely, if $\provevia{\MELL} F $, then there are some formulas $F'$, $F''$ and $F'''$  such that
  $$  \provevia {\MELL^\ell} F' 
  \provevia{\set{\dderrule,\ddigrule,\djdigrule}} F ''
  \provevia{\set{\dewrule,\dbotrule}}F'''
   \provevia{\set{\deep\ecrule}} F  
   $$
\end{theorem}
\begin{proof}
After Proposition \ref{prop:jumps}, it suffices to 
replace each occurrence of $\digrule$, $\jdigrule$, $\derrule$, $\jbotrule$, $\jewrule$,  and $\ecrule$ by an occurrence of their corresponding deep version $\ddigrule$, $\djdigrule$, $\dderrule$, $\dbotrule$, $\dewrule$,  and $\decrule$ and then to push these inferences to the bottom of the derivation\footnote{A part of these rules permutations may be performed without using deep inference rules, as shown in \Cref{fig:commCut}, but they are not enough to prove the result.}.

The converse is proven by reverting the previous argument, that is, by pushing up all deep rules applications and replace them by the corresponding non-deep rules.
\end{proof}

\section{Relation Webs}\label{sec:relweb}

\emph{Cographs} are graphs encoding formulas constructed using a conjunction and a disjunction connective~\cite{duffin:65}.
In this section we present \emph{modal relation webs}, which generalize cographs, and which will be used in the next sections to encode $\MELL$-formulas
We define modal relation webs as mixed graphs (i.e., graphs with both directed and undirected edges) satisfying certain topological conditions.
Moreover, we show that they identify formulas modulo associativity and commutativity of $\ltens$ and $\lpar$.

A \emph{directed graph} $\gG=\tuple{V_\gG,\dedge[\gG]}$ is a set
$V_\gG$ of \emph{vertices} equipped with an irreflexive binary \emph{edge relation}
$\dedge[\gG]\subseteq V_\gG\times V_\gG$.  
An
\emph{undirected graph} $\gG=\tuple{V_\gG,\uedge[\gG]}$ is a graph whose
\emph{edge relation} $\uedge[\gG]\subseteq V_\gG\times V_\gG$ is
irreflexive and symmetric. A \emph{mixed graph} is a triple
$\gG=\tuple{V_\gG,\uedge[\gG],\dedge[\gG]}$ where
$\tuple{V_\gG,\uedge[\gG]}$ is an undirected graph and
$\tuple{V_\gG,\dedge[\gG]}$ is a directed graph, such that
$\uedge[\gG] \cap \dedge[\gG]= \emptyset$ and $\dedge[\gG]$ is irreflexive.
We omit the
index/superscript $\gG$ when it is clear from the context.  
When drawing a graph we draw $\vv1\quad\vw1\Redges{v1/w1}$ whenever 
$v\uedge w$, and $\vv2\quad\vw2\modedges{v2/w2}$ whenever $v\dedge w$;
otherwise we either draw no edge at all, 
or we draw $\vv3\quad\vw3\oredges{v3/w3}$ when we want to underline the absence of edges.
A mixed graph is \emph{$\cL$-labeled} if each vertex $v$ carries a label $\lab v$ selected from a label set $\cL$. 
In this paper we fix the label set to be $\cL=\cA\cup\cneg\cA\cup\set{\oc,\wn, \lone, \lbot, \ljump}$.

\begin{definitionn}
  A \emph{relation web} is a non-empty mixed graph
  $\gG=\tuple{V_\gG,\uedge[\gG],\dedge[\gG]}$ 
 such that:
 
 \begin{itemize} 
 \item $\dedge[\gG]$ is transitive and irreflexive;
 \item $\tuple{V_\gG,\uedge[\gG]}$ is a \emph{cograph}, $\tuple{V_\gG,\dedge[\gG]}$ is a \emph{series-parallel order} and $\gG$ is \emph{3-color triangle-free}, 
 that is, $\gG$ contains  no induced subgraphs of the following shape:
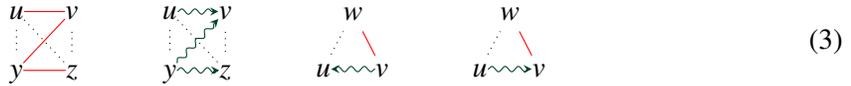
\begin{equation}\label{eq:Z-free}
	\begin{array}{c@{\hskip2em}c@{\hskip2em}c@{\hskip2em}c}
		 \begin{array}{c@{\quad\;\;}c}
			  \vu2 & \vv2\\
			  \\[-1ex]
			  \vy2 &\vz2
		 \end{array}
		 \Redges{u2/v2,y2/v2,y2/z2}
		 \oredges{u2/y2,u2/z2,z2/v2}
		&		
		\begin{array}{c@{\quad\;\;}c}
			\vu 1 & \vv 1\\
			\\[-1ex]
			\vy 1 &\vz 1
		\end{array}
		\modedges{u1/v1,y1/v1,y1/z1}
		\oredges{u1/y1,u1/z1,z1/v1}
		&		
		\begin{array}{c@{\;\;}c@{\;\;}c}
			&\vw1\\
			\\[-1.2ex]
			\vu1 &&\vv1
		\end{array}
		\oredges{u1/w1}
		\Redges{w1/v1}
		\modedges{v1/u1}	
		&	
	      \begin{array}{c@{\;\;}c@{\;\;}c}
			&\vw1\\
			\\[-1.2ex]
			\vu1 &&\vv1
		\end{array}
		\oredges{u1/w1}
		\Redges{w1/v1}
		\modedges{u1/v1}
	\end{array}
  \end{equation}
\end{itemize}
A \emph{cograph} is an undirected graph containing no induced subgraph as the leftmost one in \Cref{eq:Z-free}.
\end{definitionn}

Let $\gG$ and $\gH$ be two disjoint mixed graphs. We define the following operations:
\begin{equation}
  \label{eq:graph-operations}
  \hbox {$\hss%
  \begin{array}{c@{\;\;=\;\;}l}
    \gG\lpar\gH &
    \upsmash{\tuple{V_\gG\cup V_\gH\;,
      \;\uedge[\gG]\cup\uedge[\gH]\;,
      \;\dedge[\gG]\cup\dedge[\gH]}}\\
    \gG\lseq\gH &
    \tuple{V_\gG\cup V_\gH\;,
      \;\uedge[\gG]\cup\uedge[\gH]\;,
      \;\dedge[\gG]\cup\dedge[\gH]\cup\set{(u,v)\mid u\in V_\gG,v\in V_\gH}}\\
    \gG\ltens\gH &
    \tuple{V_\gG\cup V_\gH\;,
      \;\uedge[\gG]\cup\uedge[\gH]\cup\set{(u,v),(v,u)\mid u\in V_\gG,v\in V_\gH}\;,
      \;\dedge[\gG]\cup\dedge[\gH]}\\
  \end{array}\hskip1em
  \hss$}
\end{equation}

which can be visualized as follows:
$$
\begin{array}{c@{\qquad\;}c@{\qquad\;}c}
  \gG\lpar\gH &\gG\lseq\gH &\gG\ltens\gH \\
\begin{tikzpicture}
\node (A) at (0,0) {$\begin{array}{c}\gG \\ \vbull1\\ \vvdotsnode1 \\ \vbull2 \\~\end{array}$};
\node (B) at (2,0){$\begin{array}{c}\gH \\ \vbull3\\ \vvdotsnode2 \\ \vbull4 \\ ~\end{array}$};
\draw[-](A) circle [x radius=5mm, y radius=8mm];
\draw[-](B) circle [x radius=5mm, y radius=8mm];
\end{tikzpicture}
\multioredges{bull1,bull2,vdotsnode1}{bull4,bull3,vdotsnode2}
&
\begin{tikzpicture}
\node (A) at (0,0) {$\begin{array}{c}\gG \\ \vbull1\\ \vvdotsnode1 \\ \vbull2 \\~\end{array}$};
\node (B) at (2,0){$\begin{array}{c}\gH \\ \vbull3\\ \vvdotsnode2 \\ \vbull4 \\ ~\end{array}$};
\draw[-](A) circle [x radius=5mm, y radius=8mm];
\draw[-](B) circle [x radius=5mm, y radius=8mm];
\end{tikzpicture}
\multiMedges{bull1,bull2,vdotsnode1}{bull4,bull3,vdotsnode2}
&
\begin{tikzpicture}
\node (A) at (0,0) {$\begin{array}{c}\gG \\ \vbull1\\ \vvdotsnode1 \\ \vbull2 \\~\end{array}$};
\node (B) at (2,0){$\begin{array}{c}\gH \\ \vbull3\\ \vvdotsnode2 \\ \vbull4 \\ ~\end{array}$};
\draw[-](A) circle [x radius=5mm, y radius=8mm];
\draw[-](B) circle [x radius=5mm, y radius=8mm];
\end{tikzpicture}
\multiRedges{bull1,bull2,vdotsnode1}{bull4,bull3,vdotsnode2}
\end{array}
\vadjust{\vskip-2ex}
$$

\begin{theorem}[\cite{gug:SIS,bechet:etal:97}]\label{thm:relweb}
  A mixed graph is a relation web if and only if it can be constructed from single vertices using the three operations $\lpar$, $\lseq$ and $\ltens$ defined in~\Cref{eq:graph-operations}.
\end{theorem}

For each formula $F$ we define the $\cL$-labeled relation web $\relwebof F$.
We use the notations $\singlevertex[a]$, $\singlevertex[\cneg a]$, $\ssoc$,
$\sswn$, $\ssone$, $\ssbot$  and $\ssjump$ for the graph consisting of a single vertex
that is labeled with $a$, $\cneg a$, $\oc$, $\wn$, $\lone$, $\lbot$ and $\ljump$ respectively.
\begin{equation}
  \label{eq:translation}
  \hskip.3em
  \begin{array}{c@{\;\;=\;\;}l}
    \relwebof a & \singlevertex[a]\\
    \relwebof {\cneg a} &  \singlevertex[\cneg a]
  \end{array}
  \hskip1.5em
  \begin{array}{c@{\;\;=\;\;}l}
    \relwebof{A\ltens B}& \relwebof A\ltens\relwebof B\\
    \relwebof{A\lpar B} & \relwebof A\lpar\relwebof B
  \end{array}
  \hskip1.5em
  \begin{array}{c@{\;\;=\;\;}l}
    \relwebof{\oc A}& \ssoc\lseq\relwebof A\\
    \relwebof{\wn A} & \sswn\lseq\relwebof A
  \end{array}
  \hskip1.5em
    \begin{array}{c@{\;\;=\;\;}l}
    \relwebof{\lone}& \ssone \\
    \relwebof{\lbot} & \ssbot
    \end{array}
  \hskip1.5em
    \begin{array}{c@{\;\;=\;\;}l}
    \relwebof{\ljump}& \ssjump \\
    \end{array}
\end{equation}
For a sequent $\Gamma=A_1,\ldots,A_n$ we define $\relwebof\Gamma=\relwebof{A_1,\ldots,A_n}=\relwebof{A_1}\lpar\cdots\lpar\relwebof{A_n}$.

\begin{definitionn}\label{def:modalrelweb}
	A relation web $\gG$ is \emph{properly labeled} if for each $v,w\in\vertices[\gG]$ we have $v\dedge w$ iff $\lab   v\in\set{\oc,\wn}$.
	  
  	Moreover, we say that a relation web $\gG$ is \emph{modal} if it is properly labeled and for any vertices $u$, $v$,
	$w$ with $u\dedge w$ and $v\dedge w$ we have $u\dedge v$ or $v\dedge
	u$ or $u=v$, i.e., $\gG$ does not contain the two configurations below.
  \begin{equation}
    \label{eq:modal}
  \text{\rm Forbidden configurations for modal relation webs:}
  \hskip1em
  \upsmash{%
    \begin{array}{c@{\;\;}c@{\;\;}c}
      &\vw1\\
      \\[-1.2ex]
      \vu1 &&\vv1
    \end{array}
    \oredges{u1/v1}
    \modedges{u1/w1,v1/w1}
  \hskip1.5em
    \begin{array}{c@{\;\;}c@{\;\;}c}
      &\vw1\\
      \\[-1.2ex]
      \vu1 &&\vv1
    \end{array}
    \Redges{u1/v1}
    \modedges{u1/w1,v1/w1}
    }
  \end{equation}
\end{definitionn}

By adapting the proofs in \cite{acc:str:CPK}, we have the following results:
\begin{theorem}\label{thm:modal}
  A relation web is the translation of a formula iff
  it is modal.
\end{theorem}
\begin{proof}
	If $\gG=\relwebof{F}$ for some formula $F$, then each vertex with an outgoing $\dedge$-edge is the encoding of the modality of a subformula of $F$, hence such vertex is labeled with $\oc$ or $\wn$. 
	Moreover, if two distinct such vertices $u$ and $v$ have a $\dedge$-edge to some vertex $w$, then that $w$ is the encoding of a modality or an atom occurring in a subformula in the scope of the modalities corresponding to $u$ and $v$. Thus one of such modalities is in the scope of the other and we have $u\dedge v$ or $v\dedge u$.
	The converse follows from Theorem~\ref{thm:relweb} and the fact that the operation $\lseq$ in~\Cref{eq:graph-operations} is associative.
\end{proof}

If $\gG$ is a modal relation web, we denote by $\avertices[\gG]$, $\onevertices[\gG]$, $\jumpvertices[\gG]$, $\bvertices[\gG]$ and $\dvertices[\gG]$ the set of vertices in $\vertices[\gG]$ with labels respectively in $\atoms\cup \cneg\atoms$, $\set{\lone}$, $\set\ljump$, $\set\oc$ and $\set\wn$.
We call \emph{atomic}, \emph{unit}, \emph{jump}, and \emph{modal} vertices the ones respectively in $\avertices[\gG]$, $\onevertices[\gG]$, $\jumpvertices[\gG]$, and $\mvertices[\gG]=\bvertices[\gG]\cup \dvertices[\gG]$ .

\begin{proposition}\label{prop:relwebPoly}
  Given a set $\vertices[\gG]$ and two binary edge relations $\uedge[\gG]$ and $\dedge[\gG]$ on vertices 
  it can be checked in time polynomial on the size $\sizeof{V_\gG}$, whether $\gG=\tuple{V_\gG,\uedge[\gG],\dedge[\gG]}$ is a modal relation  web.
\end{proposition}
\begin{proof}
  Checking the transitivity and  irreflexivity of $\dedge[\gG]$ and symmetry  of $\uedge[\gG]$
  is polynomial on $\sizeof{\vertices[\gG]}$. 
  Then, to check the absence of the forbidden configurations in~\eqref{eq:Z-free} and~\eqref{eq:modal}
  we just check all triples and quadruples of vertices, which is  $\mathbf O(\sizeof{\vertices[\gG]}^4)$.
  Checking the property of being properly labeled is liner on $\sizeof{\vertices[\gG]}$.
\end{proof}

By associativity of the graph operations $\lpar$, $\ltens$ and $\lseq$ in~\Cref{eq:graph-operations} and the commutativity of  $\lpar$ and $\ltens$  we have the following equivalence.
\begin{proposition}\label{prop:formulaeq}
  For two formulas $F$ and $F'$, we have $\relwebof F=\relwebof{F'}$
  iff $F$ and $F'$ are equivalent modulo associativity of and
  commutativity of $\ltens$ and  $\lpar$.
\end{proposition}

\section{\RGB-Cographs}\label{sec:RGB}

As discussed in \Cref{subsec:handoming}, \RB-cograph are an alternative syntax for $\MLL$ proof nets.
In this paper we consider 
\RGB-cographs from \cite{acc:str:CPK},
which are an extension of the \RB-cograph syntax. 
We provide a further extension of \RGB-cographs and we establish 
a correspondence between these graphs and  $\linearized\MELL$-proofs.

\newcommand{\axlink}[1][]{\mkern1mu\mathord{\stackrel{#1}{{\gclr\curlyvee}}}\mkern1mu}
\newcommand{\varRGB}[2]{{\tuple{{\gclr #1}\mid\axlink[#2]}}}
\newcommand{\vvarRGB}[3]{\tuple{{\gclr #1}\mid\axlink[#2]\cup\axlink[#3]}}
\def\und#1{\lfloor{#1}\rfloor}
\def\mcor{$\MLL$-correct\xspace}
\def\mucor{$\MLLu$-correct\xspace}
\def\muscor{$\MELL$-correct\xspace}
\def\Xcor{$\X$-correct\xspace}

\begin{definitionn}\label{def:RGB}
  An \emph{\RGB-cograph} is a tuple
  $\cgG=\tuple{\vertices[\gG],\uedge[\gG],\dedge[\gG],\axlink[\gG]}$
  where 
  $\und{\cgG}=\tuple{\vertices[\gG],\uedge[\gG],\dedge[\gG]}$ is a modal relation web, 
  and $\axlink[\gG]$ is an equivalence relation over $\vertices[\gG]$, called the \emph{linking}, such that
  \begin{itemize}

  \item if $v\in\avertices[\gG]$ then there is exactly another $w\in\avertices[\gG]$ with $v\axlink w$ and $v\neq w$;

	\item if $v\in \onevertices[\gG]$ then $w\in \jumpvertices[\gG]$ for all $w\axlink v$ such that $w\neq v$;
  
    \item if $v\in \jumpvertices[\gG]$ then there is a  $u\in \avertices[\gG]\cup \onevertices[\gG]$ such that $w\axlink v$;

  \item if $v\in \bvertices[\gG]\cup  \dvertices[\gG]$ then there is a unique $w\in \bvertices[\gG]$ such that $w\axlink v$ and no $w\in \jumpvertices[\gG]$ such that $w\axlink v$ ;

  \end{itemize}
In particular, an \emph{\RB-cograph} is an  \RGB-cograph $\cgG$ with   $V_\gG=\avertices[\gG]$.
\end{definitionn}

These conditions can be interpreted as follows: 
the jumps vertices are associated to either to a pair of atomic vertices or to a single unit vertex; 
modal vertices are grouped in classes containing a unique vertex in $\bvertices[\gG]$.
For readers familiar with proof nets syntax, 
$\ljump$-vertices can be seen as
placeholders for the proof net jumps
while 
the $\axlink$-classes containing $\oc$ and $\wn$-vertices can be thought as encoding borders of boxes, where the unique $\oc$-vertex is the box principal ports and the $\wn$-vertices are the auxiliary ports.
More precisely, the content of the box delimited by a $\axlink$-class is the induced suggraph containing all the vertices $v$ such that $w\dedge v$ for a $w$ in the $\axlink$-class\footnote{A similar process of reconstructing boxes from the paths in the graph can be found in \cite{lamarche:essential}.}.

In drawing an \RGB-cograph we use bold (blue) edges $\vv1\quad\vw1\Bedges{v1/w1}$ when $v\neq w$ and $v\axlink w$.
We may omit edges which can be deduced by transitivity:

\def\aeA{\ae$\atoms$}

\begin{definitionn}\label{def:ae-path}
  An \emph{\ae-path} 
  in an RGB-cograph $\cgG$ is an elementary path  $x_0,x_1,\ldots,x_n$ in the graph $\tuple{V, \uedge\cup\dedge\cup \axlink}$ whose
  edges are alternating in $\axlink$ and in$\dedge\cup\uedge$.
A  \emph{chord} in an \ae-path is an edge $x_i\uedge x_j$ or $x_i\dedge
  x_j$ for 
  {$i,j\in\set{0, \dots, n}$} and $i+2\le j$. A \emph{chordless \ae-path} is
  an \ae-path without chord. An \emph{\ae-cycle} is an \ae-path 
  such that $x_0=x_n$. An RGB-cograph $\cgG$ is
  \emph{\ae-connected} if any two vertices are connected by a chordless
  \ae-path, 
  and $\cgG$ is \emph{\ae-acyclic} if it contains no   chordless \ae-cycle.
\end{definitionn}

Connectedness and acyclicity are used to define the following notions of correctness.
\begin{definitionn}\label{def:Xcorrect}
We say that  an \RGB-cograph $\cgG$ is \emph{\muscor} if it satisfies the following conditions:
  \begin{enumerate}
	\item\label{c:mll} 
	$V_\gG\neq \emptyset$ and $\cgG$ is \ae-connected and \ae-acyclic;

	\item\label{c:contextual}
    if  $w\dedge[\gG]v$ and  $v\axlink v'$, then there is  $w'\axlink w $ such that $w'\dedge[\gG]v'$.
  \end{enumerate}
We say that \RGB-cograph is \emph{\mcor} (or \emph{\mucor})
if
\emph{\muscor}
and
$\vertices = \avertices$ (respectively $\vertices = \avertices\cup \jumpvertices$ )\footnote{Note that for \mcor and \mucor \RGB-cographs condition \ref{c:contextual} is always trivially satisfied since $\bvertices\cup\dvertices=\emptyset$.}.
\end{definitionn}

\begin{figure}[t]
\begin{adjustbox}{max width=\textwidth}
	\def\jrel{j}
	$
	\begin{array}{c}
	\vlinf{}{\jaxrule}{\begin{array}{ccccc} &\va1 &\vna1 \\ \vuj 1 &\vdotsnode1 & \vdotsnode2 & \vuj2 \end{array}}{}
	\Bedges{a1/na1, a1/uj1,a1/uj2,na1/uj2,na1/uj1,uj1/dotsnode1, dotsnode1/dotsnode2, dotsnode2/uj2} 
		\hskip2em
	\vlinf{}{\jonerule}{\begin{array}{cccccc} &&\vuo 1 & \\  \vuj 1 &\vdotsnode1 && \vdotsnode2 & \vuj2 \end{array}}{}
	\Bedges{uo1/uj1,uo1/uj2, uj1/dotsnode1, dotsnode1/dotsnode2, dotsnode2/uj2}
	\hskip2em
		\vlinf{}{\lpar}{\varRGB{\gG',\gA\lpar\gB}\gG}{\varRGB{\gG',\gA,\gB}\gG}
	\hskip2em
		\vliiinf{}{\ltens}{
			\vvarRGB{\gG',\gA\ltens\gB,\gH'}\gG\gH}{
			\varRGB{\gG',\gA}{\gG}}{}{\varRGB{\gB,\gH'}{\gH}}
\\
		\vlinf{}{\krule}{
			\vvarRGB{\ssoc \lseq\gG_1, \sswn \lseq\gG_2,\dots,\sswn \lseq\gG_n}\gG\ast}{
			\varRGB{ \gG_1, \gG_2, \dots , \gG_n }{\gG}
		}
	\hskip2em
	\mbox{where }\begin{array}{l}
		\axlink[\ast]=
		\set{(v,w)\mid
			v,w\in\bvertices\uplus\dvertices \mbox{ such that }
			v,w\notin\vertices[\gG_1]\cup\cdots\cup\vertices[\gG_n]}
	\end{array}
\end{array}
	$
\end{adjustbox}
	\caption{Translating $\MELL^\ell$ sequent proofs into \RGB-cographs}
	\label{fig:MLL-RGB}
\end{figure}

\begin{lemma}
	\label{lemma:linseq}
Let $\X\in \set{\MLL, \MLLu,\MELL}$ and $F$ be a formula.
If $\provevia\Xl F$ then there is a \Xcor RGB-cograph $\cgG$ such that $\und{\cgG}=\relwebof F$.
\end{lemma}
\begin{proof}
	Let $\pi$ be a derivation of $F$ in $\X$. 
	We define a derivation of a \RGB-cograph $\cgG$ such that $\und{\cgG}=\relwebof F$ by induction on the size of $\pi$ using the rules in \Cref{fig:MLL-RGB}.
	In fact, all these rules preserve the condition in Definition~\ref{def:Xcorrect}.
\end{proof}

\begin{lemma}\label{lemma:lindeseq}
	Let $\X\in \set{\MLL, \MLLu,\MELL}$ and $\cgG$ be a  \RGB-cograph with $\und{\cgG}=\relwebof F$.
	If $\cgG $ is \Xcor, then $\provevia{\Xl} F$.
\end{lemma}
\begin{proof}

	If $\X=\MLL$, then each \mcor \RGB-cograph is an \ae-connected \ae-acyclic \RB-cograph, the result is proven in \cite{retore:phd,ret:99,retore:03}.
	The proof proceeds by induction on the size of an \RB-cograph $\cgG$ showing that each \RB-cograph is either an $\axlink$-class or there is a \emph{splitting}, that is, 
	$\vertices[\cgG]=U\uplus V$ and
	if $u\axlink[\cgG]v$ then either $u,v\in U$ or $u,v\in V$,
	and there are $U'\subset U$ and $V'\subset V$ such that for all $u\in U$ and $v\in V$, $u\uedge[\cgG]v$ iff $u\in U'$ and $u\in V'$.
	In particular, each $\axlink$-class  $\set{a,\cneg a}$ of a \mcor \RGB-cograph $\cgG$  encodes an $\axrule$-rule with conclusion $a,\cneg a$, and each splitting encodes a $\ltens$-rule.
	
	If $\X=\MLLu$, then the statement straightforwardly follows
	the the previous result.
	In this case each $\axlink$-class encodes either an $\jaxrule$-rule if it contains a pair of atomic vertices, or a $\jonerule$-rule in case it contains a unit vertex.

	If $\X=\MELL$, 
	the proof strategy is to define for each \muscor \RGB-cograph $\cgG$ a \mucor \RGB-cograph $\crbG$, and then apply the previous results. 
	Then we shall use the derivation corresponding to $\crbG$ to reconstruct a derivation for $\cgG$.

	\textbf{Definition of $\crbG$:}
	  We define $\crbG$ as follows. 
	  We define a vertex set $\nvertices=\set{v',\cneg v'\mid v\in\bvertices[\gG]\uplus\dvertices[\gG]}$ and let
	  $\vertices[\crbG]=\avertices[\cgG]\uplus\onevertices[\cgG]\uplus\jumpvertices[\cgG]\uplus\nvertices$, i.e., we replace in $\cgG$ each modal vertex by a dual pair of atomic vertices, that are linked by $\axlink[\cgG]$.
	  Moreover, the relation $\axlink[\crbG]$ is the same as in $\axlink[\cgG]$ on vertices in $\avertices[\cgG]\uplus\onevertices[\cgG]\uplus\jumpvertices[\cgG]$. 
	  In order to define $\uedge[\cgG]$, we define the following relation: 
	  $$
	  x\auxedge[\cgG]y\quad\iff\quad x\uedge[\cgG]y \mbox{~and there is no~} v\in\bvertices[\cgG]\uplus\dvertices[\cgG] \mbox{~with~}
	  x\uedge[\cgG]v\dedge[\cgG]y \mbox{~or~} y\uedge[\cgG]v\dedge[\cgG]x
	  $$
	  Now, let $x\uedge[\crbG] y$ iff one of
	  the following cases holds:
	  \begin{itemize}
	  \item $x,y\in\avertices[\gG]$ and $x\auxedge[\gG]y$;
	  \item $x\in\avertices[\gG]$ and $y=w'$ for some $w\in\bvertices[\gG]\uplus\dvertices[\gG]$ with $x\auxedge[\gG]w$;
	  \item $x=v'$ and $y=w'$ for some $v,w\in\bvertices[\gG]\uplus\dvertices[\gG]$ with $v\auxedge[\gG]w$;
	  \item $x=\cneg v'$ for some $v\in\bvertices[\gG]\uplus\dvertices[\gG]$ and $y\in\avertices[\gG]$ with $v\dedge[\gG]y$;
	  \item $x=\cneg v'$  and $y=w'$ for some $v,w\in\bvertices[\gG]\uplus\dvertices[\gG]$ with $v\dedge[\gG]w$;
	  \item $x=\cneg v'$ for some $v\in\bvertices[\gG]\uplus\dvertices[\gG]$ and $y\in\avertices[\gG]$ and there is a $u\in\bvertices[\gG]\uplus\dvertices[\gG]$ with $v\axlink[\gG]u\dedge[\gG]y$;
	  \item $x=\cneg v'$  and $y=w'$ for some $v,w\in\bvertices[\gG]\uplus\dvertices[\gG]$  and there is a $u\in\bvertices[\gG]\uplus\dvertices[\gG]$ with $v\axlink[\gG]u\dedge[\gG]w$;
	  \end{itemize}
	
	\newvertex{vpi}{{v'_i}}{}
\newvertex{nvpi}{\cneg{v'_i}}{}	
	\textbf{Properties of $\crbG$.}
	  The intuition behind this construction can be explained using
	  Theorem~\ref{thm:relweb}. 
	  Following~\cite{gug:SIS}, we use the
	  term \emph{$\BV$-formula} for an expression built from the atoms and the symbols $\oc$, $\wn$ and $\ljump$ using the binary operations $\lpar$, $\ltens$, and $\lseq$.
	  In~\cite{gug:SIS} it is shown that $\BV$-formulas, modulo associativity and commutativity of $\lpar$
	  and $\ltens$, and associativity of $\lseq$, 
	  are in one-to-one correspondence with relation webs
	  via~\eqref{eq:graph-operations} and \Cref{prop:formulaeq}.
	  We write $\cform{\gG}$ and $\crb{\form{\gG}}=\cform{\rbG}$ for a corresponding $\BV$-formula expression for respectively $\cgG$ and $\crbG$. 
	  If $v_1,\ldots,v_n\in\bvertices[\gG]\uplus\dvertices[\gG]$ form an $\axlink[\gG]$-equivalence class, this means that $\cform\gG$ is of  shape $F\cons{v_1\lseq B_1}\cdots\cons{v_n\lseq B_n}$ for some $n$-ary context $F\conhole\cdots\conhole$ (because $\cgG$ is modal). We can reformulate the translation above as follows:
	  \begin{equation}
	    \label{eq:rbtrans}
	    \hskip-1.3em
	    \rbtrans{F\cons{v_1\lseq B_1}\cdots\cons{v_n\lseq B_n}}=
	    (\cneg v'_1\ltens\cdots\ltens \cneg v'_n\ltens\rbtrans{B_1\lpar\cdots\lpar B_n})\lpar \rbtrans{F\cons{v'_1}\cdots\cons{v'_n}}
	  \end{equation}
	  We can use~\eqref{eq:rbtrans} to construct $\crbG$ from $\cgG$
	  inductively on the number of $\axlink$-classes and
	  show
	  that 
	  if $\cgG$ is an \RGB-cograph,
	  then $\crbG$ is an \RB-cograph.
	  More precisely, 
	  if $\cgG$ is \ae-connected and \ae-acyclic 
	  then $\crbG$ and each of the $\rbtrans{B_1\lpar\cdots\lpar B_n}$ determined by the $\axlink$-classes of modalities
	  are.
	For this, observe that, a priori, moving a $B_i$ out from the context could create or destroy \ae-paths. 
	However, we only claim that 
	\ae-connectedness and \ae-acyclicity are preserved, i.e., if the original \RGB-cograph is correct, then so is the one constructed via \Cref{eq:rbtrans}.
	By way of contradiction, assume the \RB-cograph in the right-hand side sequent of \Cref{eq:rbtrans} contains a chordless \ae-cycle and $\rbtrans{B_1 \lpar \cdots \lpar B_n}$ is \ae-connected.
	This chordless cycle cannot contain atoms from both $\cneg v'_1 \ltens \cdots \ltens  \cneg v'_n \ltens  \form{B_1 \lpar \cdots \lpar B_n}$ and $\form {F\{v'_1\}\cdots\{v'_n\}}$.
	If the cycle contains two vertices in $\form{B_1 \lpar \cdots \lpar B_n}$, then it must have chords, since 
	$\rbtrans{B_1 \lpar \cdots \lpar B_n}$ is connected.
	Hence, the cycle cannot contain any $v'_i$ or $\cneg v'_i$. This means that the cycle is fully contained inside the context $F\conhole\cdots \conhole$ or inside $B_1\lpar \cdots\lpar B_n$. 
	Therefore the cycle must already be present in the original \RGB-cograph. Contradiction.
	Now pick any two vertices $x'$ and $y'$ in the right-hand side sequent of \Cref{eq:rbtrans}. We show that there is a chordless \ae-path between them. 
	Let $x$ and $y$ be the corresponding vertices in the original \RGB-cograph (if $x'$ or $y'$ are one of the $ v'_i$ or $\cneg v'_i$, take the corresponding $v_i$). 
	By assumption there is a chordless \ae-path between $x$ and $y$. 
	We can recover this path in the right-hand side sequent of \Cref{eq:rbtrans}. If the original path passes through a $v_i$, then in the new graph we can pass through the new edge $\vvpi1\quad \vnvpi1 \Bedges{vpi1/nvpi1}$. The converse is proved similarly.  
	\Cref{fig:rbtrans} shows two examples of the definition of $\crbG$.

	\textbf{Sequentialization of $\cgG$ using $\crbG$.}
	 We can now piggyback 
	 on Retor\'e's proof~\cite{retore:03} of sequentialization for
	 \RB-cographs,  to produce an $\MLLK$ sequent
	 proof for $\cform\gG$.  
	 Since $\crbG$ and each of the $\rbtrans{B_1\lpar\cdots\lpar B_n}$ determined by the $\axlink$-classes of modalities
	 are 
	 \ae-connected and \ae-acyclic \RB-cograph, 
	 there is a splitting tensor in $\cform\rbG$
	 (we can remove roots $\lpar$ via the $\lpar$-rule). If this splitting
	 tensor is also present in $\cform\gG$, we can directly apply the
	 $\ltens$ rule and proceed by induction hypothesis. If it is not present in  $\cform\gG$ then it must be of shape $\cneg v'_1\ltens\cdots\ltens \cneg v'_n\ltens\rbtrans{B_1\lpar\cdots\lpar B_n}$ and be introduced by the translation in \Cref{eq:rbtrans}. Since $\crbG$ is \ae-connected, we can without loss of generality assume the the context consists only of $v_1',\ldots,v_n'$. Otherwise our tensor would not be splitting. Hence, we have
	 \vadjust{\vskip-2ex}
	  \begin{equation}\small
	    \vlderivation{
	      \vliin{}{\ltens}{v_1',\ldots,v_n',\cneg v'_1\ltens\cdots\ltens \cneg v'_n\ltens\rbtrans{B_1\lpar\cdots\lpar B_n}}{
	        \vliiin{}{\ltens}{v_1',\ldots,v_n',\cneg v'_1\ltens\cdots\ltens \cneg v'_n}{
	          \vlin{}{\axrule}{v_1',\cneg v_1'}{
	            \vlhy{}}}{
	          \vlhy{\quad\cdots\quad}}{
	          \vlin{}{\axrule}{v_n',\cneg v_n'}{
	            \vlhy{}}}}{
	        \vlpr{}{}{\qquad\rbtrans{B_1\lpar\cdots\lpar B_n}}}}
	  \end{equation}
	  whose conclusion is $\partial((v_1\lseq B_1)\lpar\cdots\lpar(v_n\lseq
	    B_n))$. Thus, we can apply the $\krule$-rule and we can proceed by induction hypothesis.

	
	Moreover, if $\jumpvertices[\gG]\neq \emptyset$, then we conclude similarly to the case of \mucor RGB-cographs.
\end{proof}

\begin{figure}[!t]
	\newcommand{\flowvertex}[2]{\mathord{%
			\tikz[remember picture,baseline=(#2\vertexcode.base)]%
			\node[inner sep=0pt](#2\vertexcode){$#1\strut$};}}
	
	\newcommand{\vCoc}[1]{\flowvertex{\ssoc}{box#1}}
	\newcommand{\vCocip}[1]{\flowvertex{\ssoc'_{#1}}{box#1}}
	\newcommand{\vCoci}[1]{\flowvertex{\ssoc_{#1}}{box#1}}
	\newcommand{\vCnoc}[1]{\flowvertex{\cneg\ssoc}{nbox#1}}
	\newcommand{\vCnocip}[1]{\flowvertex{\cneg \ssoc'_{#1}}{nbox#1}}
	
	\newcommand{\vCwn}[1]{\flowvertex{\sswn}{dia#1}}
	\newcommand{\vCnwn}[1]{\flowvertex{\cneg\sswn}{ndia#1}}
	\newcommand{\vCwni}[1]{\flowvertex{\sswn_{#1}}{dia#1}}
	\newcommand{\vCnwnip}[1]{\flowvertex{\cneg \sswn'_{#1}}{ndia#1}}
	\newcommand{\vCwnip}[1]{\flowvertex{\sswn'_{#1}}{dia#1}}
	
	\newcommand{\vCj}[1]{\flowvertex{\ssjump}{dia#1}}
	
	\newcommand{\vCa}[1]{\flowvertex{\bullet}{a#1}}
	\newcommand{\vCna}[1]{\flowvertex{\bullet}{na#1}}
	\newcommand{\vCb}[1]{\flowvertex{\bullet}{b#1}}
	\newcommand{\vCnb}[1]{\flowvertex{\bullet}{nb#1}}
	\newcommand{\vCc}[1]{\flowvertex{\bullet}{c#1}}
	\newcommand{\vCnc}[1]{\flowvertex{\bullet}{nc#1}}
	\newcommand{\vCd}[1]{\flowvertex{\bullet}{d#1}}
	\newcommand{\vCnd}[1]{\flowvertex{\bullet}{nd#1}}
	\newcommand{\vCe}[1]{\flowvertex{\bullet}{e#1}}
	\newcommand{\vCne}[1]{\flowvertex{\bullet}{ne#1}}
	
	\small
	\hbox to\textwidth{\hfill$
		\renewcommand*{\arraystretch}{2}
		\def\myskip{\hskip .8em}
		\begin{array}{c@{\myskip}c@{\myskip}c@{\myskip}c@{\myskip}c@{\myskip}c}
			\vCa1&\vCna1 & \vCj3 & \vCb3\\
			\vCoci 1&&\vCb1 &\vCnb1 &\vCc1 &\vCnc1\\
			&&\vCwni 1 &  &&\vCwni2\\
			&\vCd1&\vCnd1 & \vCoci2 &\vCe1 &\vCne1 \\
		\end{array}
		\Bedges{a1/na1, b1/nb1,c1/nc1,d1/nd1,e1/ne1}
		\Bedges{box2/dia1,box2/dia2, dia1/dia2}
		\multiRedges{b1}{na1,a1,box1}
		\Redges{nb1/c1}
		\multiRedges{nd1}{box2, nb1, c1}
		\multiRedges{e1}{nc1,dia2}
		\multiMedges{box1}{a1,na1}
		\multiMedges{box2}{c1,nb1}
		\multiMedges{dia1}{box1, a1, na1, b1}
		\multiMedges{dia2}{nc1}
		\Bedges{na1/dia3}
		\bentBedges{dia3/a1/20}
		\modedges{dia3/b3}
		\overset{\partial}\rightsquigarrow
		\renewcommand*{\arraystretch}{1.5}
		\begin{array}{c@{\myskip}c@{\myskip}c@{\myskip}c@{\myskip}c@{\myskip}c}
			\vCa1&\vCna1 & \vCj3 \\
			\vCnocip 1\\
			\vCocip 1 &&\vCb1 &\vCnb1 &\vCc1 &\vCnc1\\
			&&\vCnwnip 1 &  &&\vCnwnip2\\
			&&\vCwnip 1 &  \vCnocip2&&\vCwnip2\\
			&\vCd1&\vCnd1 & \vCocip 2 &\vCe1 &\vCne1 \\
		\end{array}
		\Bedges{a1/na1, b1/nb1,c1/nc1,d1/nd1,e1/ne1}
		\Bedges{box1/nbox1, dia1/ndia1, box2/nbox2, dia2/ndia2}
		\multiRedges{nbox2, ndia2, ndia1}{nbox2, ndia2, ndia1, box1, b1, nb1, c1, nc1}
		\Redges{b1/box1}
		\Redges{nb1/c1}
		\multiRedges{nd1}{box2}
		\multiRedges{e1}{dia2}
		\multiRedges{nbox1}{a1,na1}
		\Bedges{na1/dia3}
		\bentBedges{dia3/a1/20}
		$
		\hskip3em
		$
		\def\myskip{\hskip 1em}
		\renewcommand*{\arraystretch}{2}
		\begin{array}{c@{\myskip}c@{\myskip}c@{\myskip}c}
			\vCb1	&\vCa1&\vCna1&\vCnc1\\
			\vCnb1&\vCoc1&\vCwn1&\vCc1	
		\end{array}
		\Bedges{a1/na1, b1/nb1,c1/nc1,box1/dia1}
		\modedges{box1/a1,dia1/na1}
		\multiRedges{nb1}{a1,box1}
		\multiRedges{c1}{dia1,na1}
		\overset{\partial}\rightsquigarrow
		\renewcommand*{\arraystretch}{1.5}
		\begin{array}{c@{\myskip}c@{\myskip}c@{\myskip}c}
			&\vCa1&\vCna1\\
			\vCb1&\vCnoc 1&\vCnwn 1 & \vCnc1\\
			\vCnb1&\vCoc1&\vCwn1&\vCc1
		\end{array}
		\Bedges{a1/na1, b1/nb1,c1/nc1,box1/nbox1, dia1/ndia1}
		\multiRedges{nbox1,ndia1}{a1,na1,nbox1,ndia1}
		\Redges{nb1/box1,dia1/c1}
		$\hfill}
	\caption{
		The \RGB-cographs for $F_1=\cneg d \lpar (d \ltens \oc (\cneg b \ltens c) \lpar \cneg e \lpar (e\ltens \wn \cneg c) \lpar \wn (b\ltens \oc (a\lpar \cneg a))) \lpar \wn f$ 
		and 
		$F_2=b \lpar (\cneg b \ltens \oc a) \lpar (\wn \cneg a \ltens c) \lpar \cneg c$, and the corresponding  \RB-cographs $\rbtrans{F_1}$ and $\rbtrans{F_2}$.
	}
	\label{fig:rbtrans}
\end{figure}

We summarize the main results of this section 
by means of the following theorem:
\begin{theorem}\label{thm:MLLX}
  Let $\cgG$ be a  RGB-cograph with $\und{\cgG}=\relwebof F$ and $\X\in \set{\MLL, \MLLu,\MELL}$. Then
  \begin{equation*}
  \mbox{$F$ provable in $\Xl$ $\iff$ $\cgG$ is \Xcor}
  \end{equation*}
\end{theorem}

\section{$\MELL$-Fibrations}\label{sec:skew}

In this section we show how specific morphisms between modal relation webs, which we call  \emph{$\MELL$-fibrations}, allow us encode  $\deep\MELL$ derivations.
We here present some of the results in \cite{acc:str:CPK} where skew fibrations are meant to capture (deep) weakening-contraction derivations in modal logic.
In fact, thanks to additional definitions, we are able to restrain the these rules applications only to the ones on $\wn$-formulas (or $\lbot$).
In the syntax proposed in this paper (deep) weakening-rules cannot properly be pushed down in the derivation since the information of the jump, 
represented by the propositional constant $\ljump$, 
is firmly attached to an an axiom- or a $\lone$-rule.
Nevertheless,  we separate  the instantiation of a weakening (the jump appearing in a $\jaxrule$- or $\jonerule$) from the  weakening application (the $\dewrule$- or $\djbotrule$-rule).

\begin{definition}\label{def:skew}
	\newvertex{fv}{\cf(v)}{}
	\newvertex{fw}{\cf(w)}{}
	\newvertex{fu}{\cf(u)}{}
	\defedgetype{O}{color=orange}{}

  Let $\gG$ and $\gH$ be modal relation webs . 
  A \emph{linear fibration}
  $\cf\colon\gG\to\gH$ is a function from $\vertices[\gG]$ to   $\vertices[\gH]$ 
  such that
  \begin{enumerate}
  
  \item
   $\cf$ preserves $\uedge$ and $\dedge$, that is, 
   if $v {\ccR}_\gG w$ then $\cf(v){\ccR}_\gH \cf(w)$ for $\cR\in\set{\uedge,\dedge}$:
  \begin{equation}\label{eq:skew1}
	\begin{array}{c@{\qquad}c}
	    v\uedge[\gG]w\implies \cf(v)\uedge[\gH]\cf(w)
	 	\quand
	    v\dedge[\gG]w\implies \cf(v)\dedge[\gH]\cf(w)
		&
		\smash{
			\begin{array}{cc}
				\vv1 & \vw1
				\\[1.2em]
				\vfv1& \vfw1
			\end{array}
			\Oedges{v1/w1,fv1/fw1}
			\Sedges{v1/fv1,w1/fw1}
		}
	\end{array}
  \end{equation}
  
  \item
  
  $\cf$ has the \emph{skew-lifting property}, that is
    \begin{equation}
    \begin{array}{c@{\qquad}c}
		   \begin{array}{l}
		  	\label{eq:skewlift}
			\mbox{for every
				$v\in\vertices[\gG]$ and $w\in\vertices[\gH]$ and
				$\cR\in\set{\uedge,\dedge}$ with ~$w\, {\ccR_\gH}\, \cf(v)$~,}
			\\
			\mbox{there is a $u\in V_\gG$ such that $u\, {\ccR_\gG}\, v$ and
				$w\,\smash{\nuedge[\gH]}\, \cf(u)$ and $w\,{\ndedge[\gH]}\, \cf(u)$.}
			\end{array}
		&
		\smash{
			\begin{array}{ccc}
				\vu1 & \vv1
				\\[1.2em]
				\vfu1 & \vfv1&\vw1
			\end{array}
			\Oedges{v1/u1,fv1/fu1,w1/fv1}
			\Sedges{v1/fv1,u1/fu1}
			\bentoredges{w1/fu1/25}
		}
	\end{array}
  \end{equation}
  
  	\item 
  
  $\cf$ is \emph{modal}, that is
  \begin{equation}\label{eq:skewmod}
  	\begin{array}{c@{\qquad}c}
	  	\begin{array}{l}
	  		\mbox{if  
	  			$u\nedge[\gG] v$ and $\cf(u) \dedge[\gH] \cf(v)$, then there is a $w\in V_\gG$ such that 
	  		}
	  		\\
	  		\mbox{
	  			$w \dedge[\gG] v$ and  $\cf(u)=\cf(w)$, or  $u \dedge[\gG] w$ and  $\cf(v)=\cf(w)$.
	  		}
	  	\end{array}
  		&
	  		\begin{array}{cccc}
	  			\vu1 & &\vw1 & \vv1
	  			\\[1.2em]
	  			\vfu1 &=& \vfw 1&\vfv1
	  		\end{array}
	  		\Medges{w1/v1,fw1/fv1}
	  		\bentMedges{fu1/fv1/-15}
	  		\Sedges{v1/fv1,u1/fu1,w1/fw1}
	  		\bentoredges{u1/v1/25}
  	\end{array}
  \end{equation}
  
  \item \label{c:labelcons}
  
  $\cf$ \emph{preserves non-jump labels} and \emph{assign jumps}, that is
  \begin{equation}\label{eq:labelcond}
  	\begin{array}{c@{\qquad}c}
  		\mbox{if  $\lab v\neq \ljump$, then $\lab v=\lab{\cf(v)}$;}
  			&
  			\mbox{if $\lab v=\ljump$, then $\lab{\cf(v)}\in\set{\ljump,\lbot,\wn}$.}
  	\end{array}
  \end{equation}

  \item \label{c:weak}
  $\cf$ has the \emph{$\ljump$-domination} property, i.e.,
  \begin{equation}\label{eq:skew2}
    \begin{array}{c}
    \mbox{
    	if $w\in \vertices[\gH]\setminus  \cf(\vertices[\gG])$, then there is a $u\in \jumpvertices[\gG]$ such that     
    	$\cf(u)\dedge w$, $\lab{\cf(u)}=\wn$,
    	}
    \\
\mbox{and if $\ccR\in\set{\uedge,\dedge}$,   then $\cf(v) {\ccR_\gH} \cf(u)$ iff $\cf(v)\ccR_\gH w$. }
    \end{array}
  \end{equation}

\item\label{c:con} 

	$\cf$ has the \emph{$\wn$-domination} property, i.e.:
\begin{equation}
\begin{array}{c}
\mbox{if $\cf(v_1)=\cf(v_2)$, then there are $w_1\neq w_2$ such that $\cf(w_1)=\cf( w_2)=w$, $\lab w=\wn$ and}
\\
\begin{array}{c@{\qquad}|@{\qquad}c}
\mbox{either } 
w_1=v_1 \mbox{ and } w_2=v_2
&
\mbox{or } 
w_1\dedge v_1 \mbox{ and } w_2\dedge v_2
\end{array}
\end{array}
\end{equation}

\end{enumerate}
\end{definition}

The conditions in \Cref{def:skew} have a simple interpretation if we  identify the vertices of $\gG$ and $\gH$ with the atoms,  $\lbot$, $\ljump$, $\wn$ and $\oc$ occurring in the corresponding encoded formulas $G$ and $H$, and the and nodes in the formula trees of $G$ and $H$.
In fact, there is a $\uedge$ between two vertices iff their least common ancestor in the formula tree is a $\ltens$.
Similarly, there is a $\dedge$ from a vertex $v$ to a vertex $w$ iff $w$ is an atom in the scope of a modality $v$; more precisely,  $v$ is the least common ancestor of $v$ and $w$.
According with this remark the above conditions have the following interpretation: 
\begin{enumerate}
	\item  
	a linear fibration does not modify the least common ancestor of the corresponding nodes in the formula tree;
	
	\item 
	if the image of a vertex $u$ is a modality dominating the image of a vertex $v$, then there is a vertex $w$ with the same image of $u$ such that it is a modality dominating $v$;

	\item 
	a skew fibration can replace an internal node $n$ by a disjunction node with one child $n$ and the other child any formula tree;
	a skew fibration can attach a formula tree below a node of a modality with no child;

	\item 
	all  labels except the $\ljump$ are preserved by $\cf$. If $\lab v$, then $\lab{\cf(v)}$ may be preserved or become either a $\lbot$ or a $\wn$;

	\item 
	if a vertex in $\gH$ is not image of a vertex in $\gG$, then there is a target of a $\dedge$ with source a vertex $\cf(u)$ with such that $\lab{\cf(u)}=\wn$  and $\lab u=\ljump$;

	\item 
	 if a  vertex $v$ of $\gH$ is the image of $n>1$ distinct vertices $v_1, \dots ,v_n$  in $\gG$, then each of these vertices is an atom or a modality of a formula of the shape $\wn A$.
\end{enumerate}

\begin{figure}[!t]
  $$
  \vlinf{}{\eqrule}{\context{(A\ltens B)\ltens C}}{\context{A\ltens(B\ltens C)}}
  \qquad
  \vlinf{}{\eqrule}{\context{(A\lpar B)\lpar C}}{\context{A\lpar(B\lpar C)}}
  \qquad
  \vlinf{}{\eqrule}{\context{A\ltens B}}{\context{B\ltens A}}
  \qquad
  \vlinf{}{\eqrule}{\context{A\lpar B}}{\context{B\lpar A}}
  $$
  \caption{Deep rules for formula equivalence.}
  \label{fig:eqrules}
\end{figure}

These conditions allow us to restraint the correspondence between contractions-weakening derivations and skew fibrations \cite{hughes:pws, str:07:RTA,acc:str:CPK, acc:str:rel,ral:str:epiccube} on the formulas on the form $\wn A$.
However, since in relation web we consider formulas modulo associativity and commutativity of $\ltens$ and $\lpar$ (see \Cref{prop:formulaeq}), we need to also consider the additional (deep) rules in \Cref{fig:eqrules} taking care of this equivalence.

\begin{proposition}\label{prop:eWC}
If $\Gamma$ and $\Gamma' $ are sequents, then $\cf : \relwebof{\Gamma' } \to \relwebof{\Gamma}$ is a linear fibration iff  \; $\Gamma' \provevia{\set{\dewrule, \dbotrule, \decrule,\djcrule\deqrule}} \Gamma$. 
\end{proposition}
\begin{proof}

If $\cf$ and $\cf'$ are linear fibrations, then by definition also $\cf\circ\cf'$ is.
We then conclude by showing that for any $\rho\in\set{\dewrule,\dbotrule, \decrule,\djcrule, \deqrule}$ if $\vlinf{\rho}{}{F}{F'}$, then there is a skew fibration $\cf_\rho\colon \relwebof{F'} \to \relwebof{F}$.

If $\rho\in \set{\deqrule}$, then $\cf_\rho$ is an identity, hence a linear fibration.
If $\rho\in\set{\deep\botrule}$, then $\cf_\rho$ is an identity preserving labels with the exception of a unique $u$  such that  $\lab u= \ljump$ and $\lab{\cf(u)}=\lbot$; thus $\cf_\rho$  is a linear fibration.
If $\rho\in\set{\deep \ewrule}$, then $\cf_\rho$ is an identity over the image  of $\cf_\rho$ preserving labels with the exception of a unique $u\in \jumpvertices$ such that  $\lab{\cf(u)}=\wn$.
Moreover, any vertex $w$ in  $\relwebof{F}$ which is not image of a vertex in $\relwebof{F'}$ is dominated by $\cf(u)$, that is $\cf(v)\dedge w$.
If $\rho\in\set{\decrule}$
then $F'=\context{\wn A\lpar \wn A}$ and $F=\context{\wn A}$, $\cf_\rho$ restricted to  $\vertices[\relwebof{\context{~}}]$ is an identity.
Moreover, $\cf_\rho$ preserves $\uedge$ and $\dedge$, is modal and has the $\wn$-domination property ($\ljump$-domination is trivially satisfied).
The case if $\rho\in \set{\djcrule}$ is proven similarly.

Conversely, by the result in \cite{acc:str:CPK,ral:str:epiccube}, we know that since $\cf$ is a modal skew fibration, then $\Gamma' \provevia{\deep\Crule,\deep\Wrule, \deqrule} \Gamma$ where $\Wrule$ and $\Crule$ are the classical logic weakening and contraction rules.
To conclude the proof it suffices to show that the Conditions~\ref{c:labelcons}-\ref{c:con} restrict the application of a $\deep\Wrule$ inside the scope of a $\wn$ image of a jump-vertex (i.e. a $\dewrule$), and an application of a $\deep\Crule$ on $\wn$-formulas, that is, an application of a  $\decrule$.
\Cref{eq:skew2} ensures that if a $\deep\Wrule$ has been applied, then a vertex $v$ with $\lab v=\ljump$ has been mapped in $\cf(v)$ with  $\lab{\cf(v)}=\wn$ and the weakened formula $A$ is entirely in the scope of this $\wn$, that is, for every vertex  $w\in\vertices[\relwebof{A}]$ we have $\cf(v)\dedge w$.
If no $\deep\Wrule$ has been applied, since $\cf$ preserves non-jump labels, then either $\lab {\cf(u)}=\ljump$, or $\lab{\cf(v)}=\lbot$ -- in which case a $\deep\botrule$ has been applied.
Moreover, by  \ref{c:con} if a $\deep\Crule$ has been applied then the contracted formula $C$ is of the shape $\wn A$.
Otherwise Condition \ref{c:con} fails. In particular, if $C=\lbot$ or $C=\oc A$, the condition on labels fails; while if $C=A\lpar B$ or $C=A\ltens B$, this condition fails the condition of existence of a the vertices $w_1$ and $w_2$.
\end{proof}

\def\ddmap{\wn}
\def\dermap{\set{\deep \derrule}}
\def\digmap{\set{\deep \digrule}}

In order to capture $\derrule$ and $\digrule$ rules application, 
it suffices to adapt the results from~\cite{acc:str:CPK}.

\begin{definitionn}
We say that two vertices $v$ and $w$ in a relation web $\cgG$ are  \emph{clones} if for all $u$ with $u\neq v$ and $u\neq w$ we have $u\cR v$ iff $u\cR w$ for all $\cR\in\set{\uedge,\dedge,\bedge,\nedge}$. If $v=w$ then they are trivially clones.

  A \emph{$\ddmap$-map} is a mapping $\cf\colon\gG\to\gH$ where $\gG$
  and $\gH$ are modal relation webs, such that
  the following conditions are fulfilled:
  \begin{itemize}
  \item if $v\neq w$ and $\cf(v)=\cf(w)$, then $v$ and $w$ are clones in $\gG$, $\upsmash{v\dedge[\gG]w}$, $\lab{\cf(v)}=\lab w\in\set{\wn,\ljump}$;
  \item if $\cf(v)\neq\cf(w)$ then $v{\ccR_\gG}w$ implies
    $\cf(v){\ccR_\gH}\cf(w)$ for any
    $\cR\in\set{\uedge,\dedge,\bedge,\nedge}$;
  \item if $v\in\vertices[\gH]$ is not in the image of $\cf$ then $\lab v=\wn$ and there is a $w\in\vertices[\gH]$ with $v\dedge w$.
\end{itemize}
\end{definitionn}

\begin{proposition}\label{prop:wnfib}
Let $\Gamma$ and $\Gamma'$ be sequents. Then,
$\Gamma'\provevia{\dderrule, \ddigrule,\djdigrule}\Gamma$
iff there is a $\ddmap$-map $\cf \colon  \relwebof {\Gamma'} \to \relwebof\Gamma$.
\end{proposition}
\begin{proof}
The follows the result in \cite{acc:str:CPK} for $\set{\mathsf{4^\downarrow,t^\downarrow}}$-maps.
It suffices to remark that the modalities $\oc$ and $\wn$ of $\MELL$ behave similarly to the modalities $\lbox$  and $\ldia$ of $\Sfour$, and that $\ljump$ can be also considered as a $\ldia$ with no subformulas in its scope.
\end{proof}

We can now define fibrations capturing $\deep\MELL$-derivations.

\begin{definition}
Let $\cf\colon\gG\to\gH$ be a map between modal relation web. We say that
\begin{itemize}
	\item
	$\cf$ is an \emph{$\MELL$-fibration}  if$\cf={\ccf''}\circ{\ccf'}$ for some ${\ccf'}\colon\gG\to\gG'$ and
${\ccf''}\colon\gG'\to\gH$, where $\ccf'$ is a linear fibration and $\ccf''$ is a $\ddmap$-map;

	\item
	$\cf$ is a \emph{$\MLLu$-fibration}  if $\cf$ is a $\MELL$-fibration with $\mvertices[\gH]=\emptyset$;
	
	\item
	$\cf$ is a \emph{$\MLL$-fibration}  if $\cf$ is a bijection and $\jumpvertices[\gH]\cup \mvertices[\gH]=\emptyset$.
\end{itemize}
\end{definition}

\begin{theorem}\label{thm:mod-skew-4t} 
  Let $\Gamma$ and $\Gamma'$ be sequents, then 
 $\Gamma' \provevia{\deep\MELL}\Gamma$
 iff there is a $\MELL$-fibration 
 $\cf\colon\relwebof{\Gamma'}\to\relwebof{\Gamma}$.
\end{theorem}
\begin{proof}
As consequence of \Cref{thm:MELL-decompose}, we have that $\Gamma' \provevia{\deep\MELL}\Gamma$ iff there is a sequent $\Gamma''$ such that  $\Gamma' \provevia{\set{\dderrule, \ddigrule,\djdigrule}}\Gamma'' \provevia {\set{\eqrule,\dewrule, \dbotrule,\decrule,\djcrule}}\Gamma$.
By \Cref{prop:wnfib} there is a $\ddmap$-map ${\ccf'}\colon\relwebof{\Gamma'}\to\relwebof{\Gamma''}$  iff 
$\Gamma' \provevia{\set{\dderrule, \ddigrule,\djdigrule}}\Gamma'' $ 
and 
by \Cref{prop:eWC} there is a linear fibration 
${\ccf''}\colon\relwebof{\Gamma''}\to\relwebof{\Gamma}$ iff 
$\Gamma'' \provevia {\set{\dewrule, \dbotrule,\decrule,\djcrule,\eqrule}}\Gamma$.
We conclude since, by definition, a $\MELL$-fibration is the composition of a $\ddmap$-map and a linear fibration.

To prove the converse, we define the relation web $\gH'$ such that $\cf=\cf'\circ\cf''$ with $\cf''\colon \gH \to \gH'$ and $\cf'\colon \gH'\to \gG$.
We construct $\gH'$ from $\gH$ as follows:
\begin{itemize}
	\item for each pair of clones $v$ and $w$ with $v\dedge w$ in $\gH$ such that $\cf(v)=\cf(w)$, remove $v$;
	
	\item if $w\in \vertices[\gG]$ such that $w\dedge\cf(v)$ for some $v\in \vertices[\gH]$ and $w=\cf(u)$ for no $u\in \vertices[\gH]$, then add $w$ to $\gH$ ad all needed edges such that $w\cR[\gH'] v$ iff $w\cR[\gG]\cf(v)$ for $\cR\in \set{\uedge,\dedge}$.
\end{itemize}
We define $\cf''$ as the identity over non-clones vertices and as $\cf$ on clones vertices. We define $\cf'$ in such a way $\cf=\cf'\circ\cf''$. By definition $\cf''$ is a $\ddmap$-map and $\cf'$ is a linear fibration.
\end{proof}

\begin{proposition}\label{def:4t-skew-polynomial}
  If $\gH$ and $\gG$ are relation webs and $\cf\colon\gH\to\gG$ a function from $\vertices[\gH]$ to $\vertices[\gG]$, then it can be decided if $\cf$ is a $\MELL$-fibration in time polynomial in $\sizeof\gG+\sizeof\gH$.
\end{proposition}
\begin{proof}
We decompose $\cf=\cf'\circ\cf''$ using the same procedure used in the proof of \Cref{thm:mod-skew-4t}.
To check if $\cf''$ is a $\ddmap$-map we have to check if the  label of all vertices which are not in the image of $\cf''$ are $\wn$, and to check if the vertices with the same image are clones. This requires a quadratic time on the size of $\vertices[\gG]$.
Similarly, checking if $\cf'$ is a linear fibration is quadratic on the size of $\gH'$, which is bounded by $\sizeof{\vertices[\gH]}+\sizeof{\vertices[\gG]}$.
\end{proof}

\section{Combinatorial Proofs}\label{sec:CP}

In this section we present a combinatorial proof syntax for $\MELLj$ using the results in the previous sections.
In particular, \Cref{thm:MELL-decompose} gives us a decomposition result allowing us 
to separate in $\MELLj$ derivation, 
the linear part, that is, $\linearized{\MELL}$,
form the resource management part, that is, $\deep\MELL$.
The first part of the proof is encoded by \RGB-cographs, while the second by $\MELL$-fibrations.

\begin{definitionn}
  A map $\cf\colon\cgG\to\relwebof{F}$ from an RGB-cograph $\cgG$ to
  a the relation web $\relwebof F$ is
  \emph{allegiant} if the following conditions are satisfied:
  \begin{itemize}
  \item if $v,w\in\avertices[\gG]$ and $v\axlink[\gG]w$ then $\cf(v)$ and $\cf(w)$ are labeled by dual atoms in $\atoms$;
  \item if $v\in\vertices[\gG]\setminus\jumpvertices[\gG]$ then $\lab{\cf(v)}=\lab v$; 
  \item if $v\in\jumpvertices[\gG]$ then $\lab{\cf(v)}\in \set{\lbot, \wn}$.
  \end{itemize}
\end{definitionn}

\begin{definition}
For $\X\in \set{\MLL, \MLLu,\MELL}$, an \emph{$\X$-combinatorial proof} of a  sequent $\Gamma$ is an $\X$-fibration
  $\cf\colon\cgG\to\relwebof\Gamma$ from an {\Xcor} \RGB-cograph   $\cgG$ to the relation web of $\Gamma$.
\end{definition}

\begin{theorem}
If $F$ is a formula and $\X\in \set{\MLL, \MLLu, \MELL}$, then
$\provevia \X F$ iff there is a $\X$-combinatorial proof $\cf\colon \cgG \to \relwebof F$.
\end{theorem}
\begin{proof}
By Proposition~\ref{prop:jumps} and Theorem \ref{thm:MELL-decompose}, 
if $\provevia \X F$ then there is a formula $F'$ such that $  \provevia {\X^\ell} F' \provevia{\deep\X}F  $.
We conclude by Theorems \ref{thm:MLLX}  and \ref{thm:mod-skew-4t}.
\end{proof}

After the result in \cite{acc:str:CPK} we have that $\MELL$-combinatorial proofs represent two $\peq$-equivalent proofs with the same syntactic object.
\begin{proposition}
Let $\pi_1$ and $\pi_2$ be two derivations in $\MELLj$, then 
$\pi_1\peq\pi_2$ iff they are represented by the same $\MELL$-combinatorial proof.
\end{proposition}
\begin{proof}
It follows by the fact that $\MELLj$ can be seen as modal logic $\Sfour$ with additional axiom-like rules ($\jaxrule$ and $\jonerule$) and restricted versions of weakening and contraction rule.
Note that, in \cite{acc:str:CPK} the proof equivalence is stated in such a way rule that permutations of $\Wrule$ and $\Crule$ below $\krule$ are not allowed.
\end{proof}

In particular, there is a one-to-one correspondence between $\MLLu$-combinatorial proofs and $\MLLu$ proof nets with jumps in \cite{hughes:simple-mult}.

\begin{theorem}
For all $\X\in\set{\MLL, \MLLu,\MELL}$,  $\X$-combinatorial proofs are a sound and complete proof system for $\X$ in the sense of \cite{cook:reckhow:79}.
\end{theorem}
\begin{proof}
Let $\cgG=\tuple{\vertices[\cgG], \uedge[\cgG], \dedge[\cgG], \axlink[\cgG]}$ be a labeled mixed graph equipped with an equivalence relation, $\gH=\tuple{\vertices[\gH], \uedge[\gH], \dedge[\gH]}$ be a labeled mixed graph and $\cf\colon \cgG \to \gH$ mapping vertices of $\cgG$ to vertices of $\gH$.
Checking if the two mixed graphs are relation web is polynomial by \Cref{prop:relwebPoly} on the sizes of $\sizeof{\vertices[\cgG]}+\sizeof{\vertices[\gH]}$ and checking if $\cgG$ is an $\X$-correct  \RGB-cograph is polynomial on $\sizeof{\vertices[\cgG]}$.
By \Cref{def:4t-skew-polynomial} we have that checking if $\cf$ is a $\X$-fibration is polynomial on $\sizeof{\vertices[\cgG]}+\sizeof{\vertices[\gH]}$.
Then, since checking  if $\cf$ is allegiant is linear on $\sizeof{\vertices[\cgG]}+\sizeof{\vertices[\gH]}$, we conclude the proof.
\end{proof}

\section{Handsome Proof Nets for $\MELL$}\label{sec:handsome}

\def\CPset{\mathsf{CP}}
\def\cutcomp#1{\circ_{\cutr}^{#1}}
\newcommand{\promotion}[2]{{\gclr \krule}[#1]_{\oc#2}}
\newcommand{\remclass}[2]{({#1} \setminus {\gclr #2})}
\newcommand{\mergeclass}[5]{{#1}{[#2\axlink[#4,#5]#3]}}
\newcommand{\restricted}[2]{#1|_{#2}}
\def\emptywf#1{{\color{\skewcolor}\emptyset}_\relwebof{#1}}
\newcommand{\idskew}{{\color{\skewcolor} {\id}}}
\newcommand{\jskew}[1]{{\color{\skewcolor}{0}_#1}}
\newcommand{\cpind}[2]{\cf \colon \cgG_{#1} \to \relwebof{#2}}
\newcommand{\norcpind}[2]{\hat{\cf} \colon \cgG_{#1} \to \relwebof{#2}}
\newcommand{\cpindi}[3]{\cf_{#3} \colon \cgG_{#1} \to \relwebof{#2}}
\newcommand{\skewminus}[2]{(#1 \setminus {\fclr #2})}
\newcommand{\skewf}[2]{(\fclr f_{{#1}\flowto {#2}})}

\newcommand{\jumpreplace}[2]{#1[#1/\ljump]}

\newcommand{\rgbcog}[1][]{\tuple{\vertices[#1], \uedge[#1],\dedge[#1],\axlink[#1]}}
\newcommand{\wnreaplacement}[3]{#1[#3/#2]}
\def\digfib#1{{\color{\skewcolor} \mathsf{dig}}({\wn\wn #1})}
\def\replaceeq#1{[{\gclr #1}]}
\def\minusfib#1{\setminus\set{{\fclr #1}}}
\def\minuseq#1{\setminus{{\gclr #1}}}
\def\skcomp#1{{\circ \fclr f_{#1}}}
\newcommand{\sfto}[2]{\relwebof{#1}\to{\relwebof{#2}}}
\newcommand{\cpp}{\tuple{{\gclr \gG}, \cf}}
\newcommand{\cppair}[2]{\tuple{{\gclr \gG_{#1}}, \skewf{#1}{#2}}}

\def\hpn{\mathsf{HPN}}

The combinatorial proofs defined in the previous section can be interpreted as an extension of both Retor\'e's \cite{retore:03} and Hughes' \cite{hughes:pws} syntaxes.
They only allow to represent $\cutr$-free proofs as \emph{unfolded} Retor\'e's proof nets \cite{ret:99} do. 
In this section we define \emph{exponentially handsome proof nets} (or \emph{$\MELL$-combinatorial proofs with cuts}) by extending $\MELL$-combinatorial proofs.
We conclude by providing a cut-elimination procedure for exponentially handsome proof nets.

We adapt the solution proposed in \cite{hughes:pws,hughes:invar} where instances of $\cutr$ are replaced by conjunction rules producing \emph{contradictions}, i.e. formulas of the form $A\ltens \cneg A$, which are later ``discarded'' at the end of the derivation.
In the classical setting, this substitution creates no further interaction of the contradiction with any other rule in the derivation.
The presence of the weak promotion rule prevents us from ignoring the presence of these contradictions, but we here propose a solution for this problem.
\def\SKS{\mathsf{SKS}}
\begin{remark}
	Another extension of combinatorial proof for classical logic allowing to represent proofs with cuts is given in \cite{str:FSCD17} by means of \emph{combinatorial flows}.
	In this setting $\cutr$ is in some way encoded as sequential composition of flows.
	This syntax is strongly related with the property of the deep inference system $\SKS$ for classical logic \cite{gug:SIS} which states that if $\provevia{\SKS}\cneg A\lpar B$, then $A\provevia{\SKS}B$.
	Such a result is out of the scope of this paper, but it could be obtained by defining a deep inference system  for $\MELLj$ based on the systems in \cite{str:phd,str:MELL}.
\end{remark}

\def\scon{weaken-contradiction\xspace}
\def\scons{weaken-contradictions\xspace}
\begin{definition}
	If $F$ is a formula, we write $\wn^n F$ to denote the formula $\smash{\wn^n F=\overbrace{\wn\cdots \wn}^{\mbox{\tiny$n$ times }} F}$. In particular, $\wn^0 F=F$.
	A \emph{contradiction} is a formula of the form $C=A\ltens \cneg A$ for a $\MELL$-formula $A$, and a \emph{\scon} is a formula of the form $\wn^n C $ with $C$ a contradiction.
\end{definition}

\begin{theorem}
	Let $\Gamma$ be a sequent and $C_1, \dots, C_n$ \scons.
	If 
	$\provevia{\MELL}\Gamma,C_1, \dots, C_n$
	then
	$\provevia{\MELL}\Gamma$. 
	Conversely, if $\Gamma$ admits a derivation in ${\MELL}\cup\set{\cutr}$ containing $k$ occurrences of the $\cutr$-rule with pairs of active formulas $(A_1,\cneg A_1), \dots, (A_k,\cneg A_k) $ then 
	$\provevia{\MELL}\Gamma,\wn^{n_1}(A_1\ltens \cneg A_1), \dots, \wn^{n_k}(A_i\ltens \cneg A_i)$
	for some $n_1, \dots, n_k\in \N$.
\end{theorem}
\begin{proof}
By induction on the number of \scons formulas $k$.
If $\Gamma,C$ is provable in $\MELL$ and $C=\wn^n(A\ltens \cneg A)$, then in any derivation of $\Gamma,C$ there is an occurrence of a rule $\ltens$ with principal formula $A\ltens \cneg A$ and active formulas $A$ and $\cneg A$. 
We can obtain a derivation of $\Gamma$ by replacing such an occurrence with a $\cutr$ having the same active formulas, and then removing each occurrence of the active formula $A\ltens \cneg A$ from the derivation.
$$
\vliinf{}{\ltens}{\Gamma, A\ltens\cneg A, \Delta}{\Gamma, A}{\cneg A, \Delta}
\qquad
\leftrightsquigarrow
\qquad
\vliinf{}{\cutr}{\Gamma, \Delta}{\Gamma, A}{\cneg A, \Delta}
%
$$
Then we conclude thanks to \Cref{thm:seq:cutelim}.

The converse is proven by similarly. 
In this case $n_i$ is defined as the number of occurrences of $\krule$ below the $\cutr$ with active formulas formula $A_i$ and $\cneg A_i$ for each $i\in\set{1,\dots k}$.
\end{proof}

\begin{definition}
An \emph{exponentially handsome proof net} of a sequent $\Gamma$ is a $\MELL$-combinatorial proof $\cf\colon {\cgG} \to \relwebof{\Gamma, \Delta}$ where $\Delta$ is a (possibly empty) sequent of \scons.
\end{definition}

In drawing such objects, we shade in gray the backgrounds of the portion of the relation web of $\Gamma$ containing the vertices of the \scons (for a graphical example refer to \Cref{fig:HPNexample}).


\def\vs{\mbox{-VS-}}
\newcommand{\redCase}[2]{\overset{#1\vs#2}{\rightsquigarrow_\cutr}}
\newvertex{p}{~}{}


We conclude this section by defining a graphical rewriting allowing us to produce a $\MELL$-combinatorial proof from an exponentially handsome proof net.

\def\energy#1{\mathsf{W}_#1}
\begin{theorem}\label{thm:EHPN:cutelim}
If $\cf\colon {\cgG} \to \relwebof{\Gamma, \Delta}$ is an exponentially handsome proof nets of $\Gamma$, then there is a $\MELL$-combinatorial proof $\cf' \colon {\cgG} \to \relwebof{\Gamma}$.
\end{theorem}
\begin{proof}
	If $\Delta= \emptyset$, then $\cf $ is a $\MELL$-combinatorial proof.
	Otherwise, in order prove this result we implement cut-elimination by translating the rewriting in \Cref{fig:cutElim}.
	We consider the following key cases which are shown in \Cref{fig:handsomeCutElim}:
		\begin{itemize}
			\item $\redCase{\axrule}{\axrule}$: 
			remove the two vertices in the contradiction and merge their $\axlink$-classes;
			
			\item $\redCase{\bot}{\lone}$: 
			similarly to the previous one;

			\item $\redCase{\ltens}{\lpar}$: 
			remove the $\uedge$-edges between the vertices in $B$ and the ones in $A$ and $\cneg A$, and remove the $\uedge$-edges between the vertices in $A$ and the ones in $\cneg B$;
			
			\item $\redCase{\krule}{\krule}$:
			replace the $\uedge$-edges between $\wn $ and a vertex in $A$ with $\dedge$-edge from $\wn$ to a vertex in $A$.
			Merge the $\axlink$-classes of $\oc$ and $\wn$ and then remove the vertex $\oc$;
			
			\item $\redCase{\krule}{\derrule}$:
			remove the $\oc$ and the vertices in its $\axlink$-class;
			
			\item $\redCase{\krule}{\jewrule}$:
			remove the $\oc$ and  any vertex pointed by the $\wn$ in its $\axlink$-class.
			Then replace each $\wn$ in with a new $\ljump$ having the same image by $\cf$.
			Add to the $\axlink$-class of the $\ljump^\cutr$ which is in the pre-image of the $\wn$ of the cut-formula the new $\ljump$-vertices.
			Finally remove $\ljump^\cutr$.
			
			\item $\redCase{\krule}{\ecrule}$:
			remove any $\uedge$ connecting $\oc A$ to any copy of $\wn \cneg A$;
			for each pre-mage of $\wn$ excepting one, make a copy of the connected component of the \RGB-cograph containing the $\axlink$-class of $\oc$ 
			for each $\wn$ in the pre-image of the $\wn$ of the cut-formula $\wn\cneg A$.
			add the $\uedge$-edges between each copy of $\oc A$ and the corresponding copy of $\wn \cneg A$;
			
			\item $\redCase{\krule}{\digrule}$:
			remove any $\uedge$ connecting $\oc A$ to any $\wn$ with image the $\wn$ of the cut-formula; 
			make a copy of the $\axlink$-class of $\oc$ for each $\wn$ in the pre-image of the $\wn$ of the cut-formula $\wn\cneg A$.
			Add the $\uedge$-edge and the $\dedge$-edges in such a way each new copy of $\wn$ and $\oc$ is a clone of the original one.
			Add a $\uedge$ between each new copy of $\oc$ and the corresponding $\wn$.
	\end{itemize}
	Note that proof equivalence and the permutations in \Cref{fig:commCut} are already captured by the syntax.

	Let $\Delta=C_1,\dots, C_n$.
	We define the function $\energy\Delta$ associating to an exponentially handsome proof nets $\cg\colon {\cgH}\to \relwebof{\Sigma}$ the string $\energy{\Delta}(\cg)=(\sizeof{\Sigma}_{C_1}, \dots, \sizeof{\Sigma}_{C_n})$ where $\sizeof{\Sigma}_{C}$ is the number of occurrences of the the formula $C$ in $\Sigma$.
	The reduction strategy consists in selecting and applying reduction to one \scon in $\cf$ at a time, selecting the \scons according with the following the order on $\Delta$: we do not select a \scon $C_{i+1}$ until we have not removed all the occurrences of $C_i$.
	At the end of each sub-routine removing an occurrence of a \scon $C_i$, we may have created new copies of the \scons, but only copies of the $C_j$ with $j>i$.
	The reduction order is the lexicographic order on the pair $(\sizeof{\vertices[\relwebof{C}]},\energy\Delta(\cg))$ where $C$ is the currently selected \scon in the elimination process.
	Since each rewriting rule in \Cref{fig:handsomeCutElim} either reduces $\sizeof{\relwebof{C}}$ or $\energy{\Delta}$, then the cut-elimination procedure ends.
\end{proof}
\begin{figure}[!t]
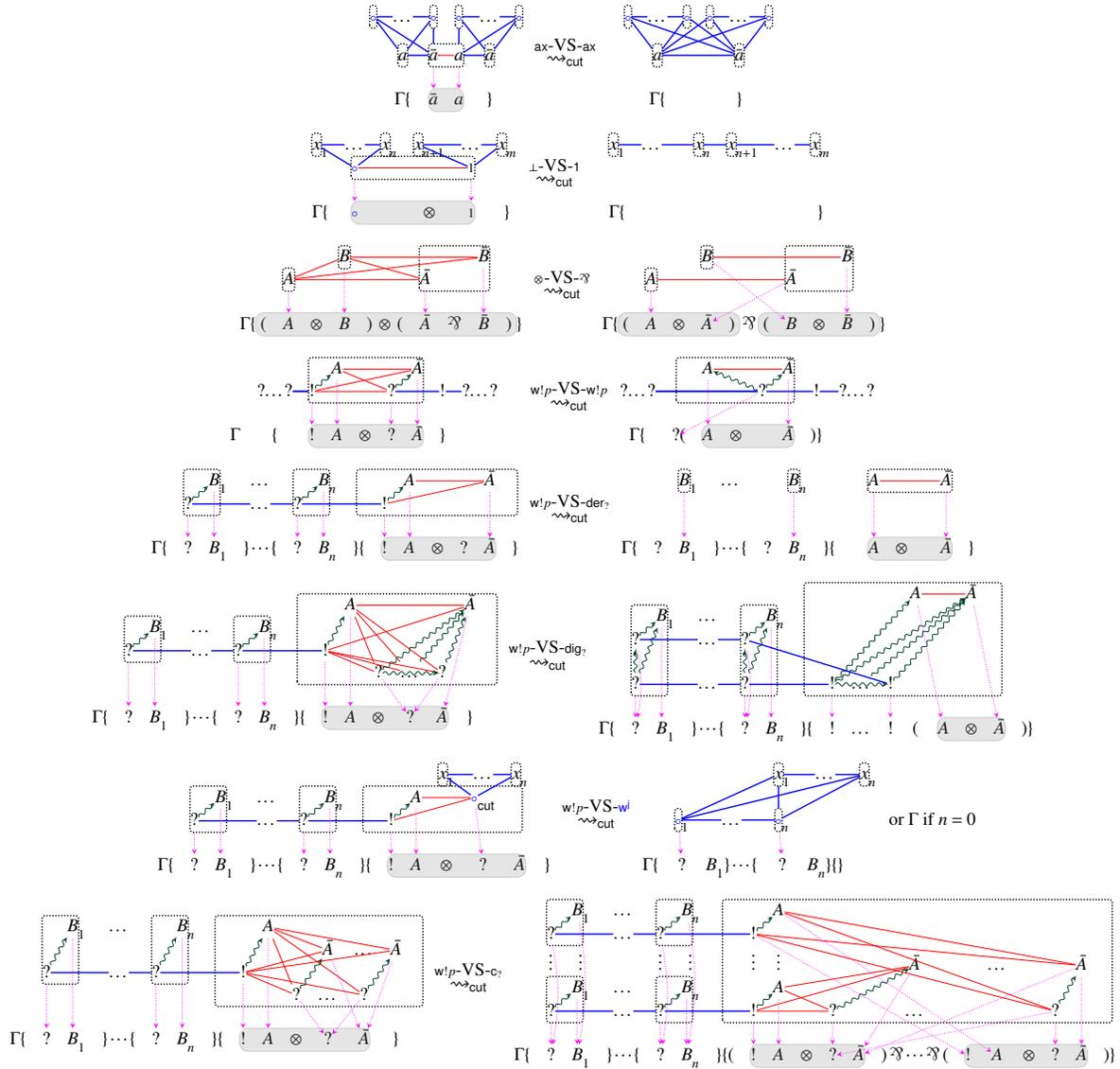

	\begin{adjustbox}{max width=\textwidth}
		$
		\begin{array}{cc}
			\begin{array}{cccccccc}
				&\vuj 1 & \vdotsnode 1 &\vuj 2 & \vuj 3 & \vdotsnode 4 &\vuj 5 
				\\[2ex]
				&\vp1		&\va1			&\vna1 &\va2 &\vna2	&
				\\
				\\
				&&\Gamma \{& \vna3 &  \va4 &\}
			\end{array}
			\Bedges{uj1/dotsnode1,dotsnode1/uj2}
			\multiBedges{a1,na1}{na1,uj1,uj2}
			\Bedges{uj3/dotsnode4,dotsnode4/uj5}
			\multiBedges{na2,a2}{a2,uj3,uj5}
			\cutshade{na3}{a4}
			\Sedges{na1/na3,a2/a4}
			\contextnodes{a1,na2,uj1,uj2,uj3,uj5}
			\Redges{na1/a2}
			\redCase{\axrule}{\axrule}
			\begin{array}{cccccccc}
				&\vuj 1 & \vdotsnode 1 &\vuj 2 & \vuj 3 & \vdotsnode 4 &\vuj 5 
				\\[2ex]
				&\vp1		&\va1			& & &\vna2	&
				\\
				\\
				&&\Gamma \{&  &   &\}
			\end{array}
			\Bedges{uj1/dotsnode1,dotsnode1/uj2}
			\multiBedges{a1,na2}{uj1,uj2,uj3,uj5}
			\Bedges{uj3/dotsnode4,dotsnode4/uj5}
			\contextbox{na1}{a2}
			\Bedges{a1/na2}
			\contextnodes{a1,na2,uj1,uj2,uj3,uj5}
			\\
			\\
			\begin{array}{ccccccccc}
				&\vx1_1 & \vdotsnode 1&\vx2_n & \vx3_{n+1} & \vdotsnode 4 &\vx5_m 
				\\
				&\vp1		&	\vuj9		&&  &\vuo1	&
				\\
				\\
				&\Gamma \{& \vuj1 &&\ltens&  \vuo2 &\}
			\end{array}
			\Bedges{x1/dotsnode1,dotsnode1/x2}
			\multiBedges{uj9}{x1,x2}
			\Bedges{x3/dotsnode4,dotsnode4/x5}
			\multiBedges{uo1}{x3,x5}
			\contextbox{uj9}{uo1}
			\cutshade{uj1}{uo2}
			\Redges{uj9/uo1}
			\Sedges{uj9/uj1,uo1/uo2}
			\contextnodes{x1,x2,x3,x5}
			\redCase{\lbot}{\lone}
			\begin{array}{cccccccc}
				&\vx1_1 & \vdotsnode 1& &\vx2_n & \vx3_{n+1} & \vdotsnode 4 &\vx5_m 
				\\
				&		&			&& & &	&
				\\
				\\
				&\Gamma \{&  &&&&   &\}
			\end{array}
			\Bedges{x1/dotsnode1,dotsnode1/x2,x2/x3,x3/dotsnode4,dotsnode4/x5}
			\contextnodes{x1,x2,x3,x5}
			\\
			\\
			\begin{array}{cccccccccccc}
				&&&\vB1			&&& &\vnB1
				\\
				&\vA1&&&& \vnA1
				\\
				\\
				\Gamma \{\vp1(&\vA2&\ltens& \vB2 &) \vp2\ltens\vp3( & \vnA2& \lpar& \vnB2&)\vp4\}
			\end{array}
			\contextbox{A1}{A1}
			\contextbox{B1}{B1}
			\contextbox{nA1}{nB1}
			\Sedges{A1/A2,nA1/nA2,B1/B2,nB1/nB2}
			\multiRedges{A1,B1}{A1,nA1,nB1}
			\cutshade{p1}{p4}
			\redCase{\ltens}{\lpar}
			\begin{array}{cccccccccccc}
				&&&\vB1			&&& &\vnB1
				\\
				&\vA1&&&& \vnA1
				\\
				\\
				\Gamma \{\vp1(&\vA2&\ltens& \vnA2 &)\vp2 \lpar \vp3( & \vB2& \ltens& \vnB2&)\vp4\}
			\end{array}
			\contextbox{A1}{A1}
			\contextbox{B1}{B1}
			\contextbox{nA1}{nB1}
			\Sedges{A1/A2,nA1/nA2,B1/B2,nB1/nB2}
			\Redges{A1/nA1,B1/nB1}
			\cutshade{p1}{p2}
			\cutshade{p3}{p4}
			\\
			\\
			\begin{array}{cccccccccccc}
				&&	& \vA2 &&&\vnA2&&&\vp1
				\\
				\vp0	&\vwn4 \vdotsnode1 \vwn3 &\voc2 & & &\vwn2 &&\voc3 & \vwn 4 \vdotsnode2 \vwn{5}
				\\
				\\
				\Gamma &\{ &\voc1 & \vA1 &\ltens &\vwn1 & \vnA1& \}
			\end{array}
			\cutshade{oc1}{nA1}
			\Sedges{A2/A1,nA2/nA1,oc2/oc1,wn2/wn1}
			\contextbox{oc2}{nA2}
			\Bedges{wn3/oc2}
			\Bedges{wn3/oc2,wn2/oc3,oc3/wn4}
			\modedges{oc2/A2,wn2/nA2}
			\multiRedges{A2,oc2}{nA2,wn2}
			\redCase{\krule}{\krule}
			\begin{array}{cccccccccccc}
				&	& \vA2 &&&\vnA2&&&\vp1
				\\
				\vwn4 \vdotsnode1 \vwn3 & \vp7& & &\vwn2 &&\voc3 & \vwn 4 \vdotsnode2 \vwn{5}
				\\
				\\
				\Gamma \{ &\vwn1 (& \vA1 &\ltens &\vp2 & \vnA1&) \}
			\end{array}
			\contextbox{p7}{nA2}
			\Bedges{wn3/wn2,wn2/oc3,oc3/wn4}
			\cutshade{A1}{nA1}
			\Sedges{A2/A1,nA2/nA1,wn2/wn1}
			\Bedges{wn3/wn2}
			\modedges{wn2/A2,wn2/nA2}
			\Redges{A2/nA2}
			\\
			\\
			\begin{array}{cccccccccccccc}
				&& \vB1_1 &	\dots 		&& \vB2_n & && \vA2 &&&\vnA2&\vp1
				\\
				\vp0	&\vwn3 &&\vdotsnode1 &\vwn2 & &\vp6&\voc 2 
				\\
				\\
				\Gamma \{ &\vwn6 &\vB3_1&\}\cdots\{&\vwn4 &\vB4_n&\} \{& \voc1& \vA1 &\ltens &\vwn9 & \vnA1& \}
			\end{array}
			\Bedges{wn3/dotsnode1,wn2/dotsnode1,oc2/wn2}
			\multiRedges{oc2,A2}{nA2}
			\Sedges{A2/A1,nA2/nA1,oc2/oc1}
			\Sedges{B1/B3,wn3/wn6,wn2/wn4,B2/B4}
			\cutshade{oc1}{nA1}
			\contextbox{wn3}{B1}
			\contextbox{wn2}{B2}
			\contextbox{p6}{p1}
			\modedges{wn3/B1,wn2/B2,oc2/A2}
			\redCase{\krule}{\derrule}
			\begin{array}{cccccccccccccc}
				&& \vB1_1 &	\dots 		&& \vB2_n & && \vA2 &&&\vnA2
				\\
				\vp0	
				\\
				\\
				\Gamma \{ &\vwn6 &\vB3_1&\}\cdots\{&\vwn4 &\vB4_n&\} \{& & \vA1 &\ltens & & \vnA1& \}
			\end{array}
			\multiRedges{A2}{nA2}
			\Sedges{A2/A1,nA2/nA1}
			\Sedges{B1/B3,B2/B4}
			\cutshade{A1}{nA1}
			\contextbox{A2}{nA2}
			\contextbox{B1}{B1}
			\contextbox{B2}{B2}
			\\
			\\
			\begin{array}{ccccccccccccccccc}
				&&&	&&& && \vA2 &&&&\vnA2&\vp1
				\\
				&& \vB1_1 &	\dots 		&& \vB2_n 
				\\
				&\vwn3 &&\vdotsnode1 &\vwn2 & &&\voc 2  
				\\
				& & && & &\vp6&  &&\vwn5 &\vdotsnode2 &\vwn7
				\\
				\\
				\Gamma \{ &\vwn6 &\vB3_1&\}\cdots\{&\vwn4 &\vB4_n&\} \{& \voc1& \vA1 &\ltens &\vwn9 & \vnA1& \}
			\end{array}
			\Bedges{wn3/dotsnode1,wn2/dotsnode1,oc2/wn2}
			\multiRedges{oc2,A2}{nA2,wn5,dotsnode2,wn7}
			\Sedges{A2/A1,nA2/nA1,oc2/oc1}
			\Sedges{B1/B3,wn3/wn6,wn2/wn4,B2/B4}
			\cutshade{oc1}{nA1}
			\contextbox{wn3}{B1}
			\contextbox{wn2}{B2}
			\contextbox{p6}{p1}
			\modedges{wn3/B1,wn2/B2,oc2/A2}
			\modedges{wn5/wn7,wn7/nA2,wn5/nA2,dotsnode2/nA2}
			\Sedges{wn5/wn9,wn7/wn9}
			\redCase{\krule}{\digrule}
			\begin{array}{ccccccccccccccccc}
				&&&	&&& && && \vA2&&\vnA2 &\vp1
				\\
				&& \vB1_1 &	\dots 		&& \vB2_n 
				\\
				&\vwn{31} &&\vdotsnode{11} &\vwn{21} & &&
				\\
				\\
				&\vwn3 &&\vdotsnode1 &\vwn2 & & \vp6 &\voc5 &\vdotsnode2 &\voc7
				\\
				\\
				\Gamma \{ &\vwn6 &\vB3_1&\}\cdots\{&\vwn4 &\vB4_n&\} \{& \voc9&\dots &\voc{10}& (&\vA1 &\ltens & \vnA1&) \}
			\end{array}
			\Bedges{wn3/dotsnode1,wn2/dotsnode1,oc5/wn2}
			\Sedges{A2/A1,nA2/nA1}
			\Sedges{B1/B3,wn3/wn6,wn2/wn4,B2/B4}
			\cutshade{A1}{nA1}
			\contextbox{wn3}{B1}
			\contextbox{wn2}{B2}
			\modedges{wn3/B1,wn2/B2,oc5/A2}
			\modedges{wn31/B1,wn21/B2,wn3/wn31,wn2/wn21}
			\modedges{oc5/oc7,oc7/nA2,oc5/nA2,dotsnode2/nA2}
			\Sedges{oc5/oc9,oc7/oc10}
			\multiRedges{A2}{nA2}
			\bentSedges{wn31/wn6/10,wn21/wn4/10}
			\Bedges{wn31/dotsnode11,dotsnode11/wn21,wn21/oc7}
			\contextbox{p6}{p1}
			\\
			\\			
			\begin{array}{cccccccccccccc}
				&&&&&&&&&\vx1_1 & \vdotsnode 6&\vx2_n
				\\
				&& \vB1_1 &	\dots 		&& \vB2_n & && \vA2 &&\vuj2_\cutr&\vp2
				\\
				\vp0	&\vwn3 &&\vdotsnode1 &\vwn2 & &\vp6&\voc 2 
				\\
				\\
				\Gamma \{ &\vwn6 &\vB3_1&\}\cdots\{&\vwn4 &\vB4_n&\} \{& \voc1& \vA1 &\ltens &\vwn9 & \vnA1& \}
			\end{array}
			\Bedges{wn3/dotsnode1,wn2/dotsnode1,oc2/wn2}
			\multiRedges{oc2,A2}{uj2}
			\Sedges{A2/A1,uj2/wn9,oc2/oc1}
			\Sedges{B1/B3,wn3/wn6,wn2/wn4,B2/B4}
			\cutshade{oc1}{nA1}
			\contextbox{wn3}{B1}
			\contextbox{wn2}{B2}
			\contextbox{p6}{p2}
			\modedges{wn3/B1,wn2/B2,oc2/A2}		
			\multiBedges{uj2}{x1,x2}\Bedges{x1/dotsnode6,dotsnode6/x2}
			\contextnodes{x1,x2}
			\redCase{\krule}{\jewrule}
			\begin{array}{cccccccccccccc}
				&&&\vx1_1 & \vdotsnode 6&\vx2_n
				\\
				\\
				&\vuj3_1 &	\vdotsnode1 	&\vuj2_n
				\\
				\\
				\Gamma \{ &\vwn6 &\vB3_1\}\cdots\{&\vwn4 &\vB4_n\} \{  \}
			\end{array}
			\Sedges{uj3/wn6,uj2/wn4}
			\contextnodes{uj2,uj3}
			\multiBedges{uj2,uj3}{x1,x2}\Bedges{x1/dotsnode6,dotsnode6/x2}
			\Bedges{uj3/dotsnode1,dotsnode1/uj2}
			\contextnodes{x1,x2}
			\mbox{ or }
			\Gamma
			\mbox{ if $n=0$}
			\\
			\\
			\begin{array}{ccccccccccccccc}
				&& \vB1_1 &	\dots 		&& \vB2_n & && \vA2 && &&&\vp7
				\\
				&&  &	 		&&  & && &&\vnA2 & \dots&\vnA3
				\\
				& \vwn3 & &\vdotsnode1&\vwn2  & &&\voc 2
				\\
				& & && & &\vp6& &&  \vwn8&\dots &\vwn7
				\\
				\\
				\Gamma \{ &\vwn6 &\vB3_1&\}\cdots\{&\vwn4 &\vB4_n&\} \{& \voc1& \vA1 &\ltens &\vwn9 & \vnA1& \}
			\end{array}
			\Bedges{wn3/dotsnode1,wn2/dotsnode1,oc2/wn2}
			\multiRedges{oc2,A2}{nA2,wn8,wn7,nA3}
			\Sedges{A2/A1,nA2/nA1,oc2/oc1,nA3/nA1,wn7/wn9,wn8/wn9}
			\Sedges{B1/B3,wn3/wn6,wn2/wn4,B2/B4}
			\cutshade{oc1}{nA1}
			\contextbox{wn3}{B1}
			\contextbox{wn2}{B2}
			\contextbox{p6}{p7}
			\modedges{wn3/B1,wn2/B2,oc2/A2}
			\modedges{wn8/nA2,wn7/nA3}
			\redCase{\krule}{\ecrule}
			\begin{array}{cccccccccccccccccccc}
				&& \vB{11}_1 &	\dots 		&& \vB{21}_n & && \vA{21} &&&&&&&&&&\vp7
				\\
				& \vwn{31} & &\vdotsnode{11}&\vwn{21}  & &&\voc {21} &
				\\
				& \vdots &\vdots &&\vdots  & \vdots&&\vdots &\vdots& &&&\vnA2&&\dots &&&\vnA3
				\\
				&& \vB1_1 &	\dots 		&& \vB2_n & && \vA2 && &&
				\\
				& \vwn3 & &\vdotsnode1&\vwn2  & &\vp6&\voc 2&&&  \vwn8&&\dots&&&&\vwn7
				\\
				\\
				\Gamma \{ &\vwn6 &\vB3_1&\}\cdots\{&\vwn4 &\vB4_n&\} \{(& \voc1& \vA1 &\ltens &\vwn9 & \vnA1&)\lpar \cdots\lpar(& \voc{11} & \vA{11} & \ltens & \vwn{17} & \vnA{13}&) \}
			\end{array}
			\Bedges{wn3/dotsnode1,wn2/dotsnode1,oc2/wn2}
			\multiRedges{oc2,A2}{nA2,wn8}
			\Sedges{A2/A1,nA2/nA1,oc2/oc1,nA3/nA1,wn7/wn9,wn8/wn9}
			\Sedges{B1/B3,wn3/wn6,wn2/wn4,B2/B4}
			\cutshade{oc1}{nA1}
			\contextbox{wn3}{B1}
			\contextbox{wn2}{B2}
			\contextbox{p6}{p7}
			\modedges{wn3/B1,wn2/B2,oc2/A2}
			\modedges{wn8/nA2,wn7/nA3}
			\bentSedges{wn31/wn6/10,B11/B3/10,B21/B4/10,wn21/wn4/10}
			\multiRedges{A21,oc21}{wn7,nA3}
			\Sedges{A21/A11,oc21/oc11,wn7/wn17,nA3/nA13}
			\modedges{wn31/B11,wn21/B21,oc21/A21}
			\Bedges{wn31/dotsnode11,wn21/dotsnode11,oc21/wn21}
			\cutshade{oc11}{nA13}
			\contextbox{wn31}{B11}
			\contextbox{wn21}{B21}
		\end{array}
		$
	\end{adjustbox}
	\captionsetup{singlelinecheck=off}
	\caption[]{Cut-elimination steps for exponentially handsome proof nets. 
		Dotted lines delimit modules, that is, sets of vertices having the same $\uedge$- and $\dedge$-relation with any vertex outside.
		To simplify the reading we write with the same symbol a formula and its pre-image by $\cf$.
	}
	\label{fig:handsomeCutElim}
\end{figure}

\section{Conclusions}

In this paper we extended Retor\'e's \RB-cograph syntax for multiplicative proof nets \cite{retore:03} in order to include units and exponentials, 
refining the combinatorial proofs for modal logic from \cite{acc:str:CPK}.

Aware of the limits in designing a syntax able to capture proof equivalence and with a polynomial correctness criterion \cite{heijltjes:houston:14}, we restrained the notion of proof equivalence 
by introducing the proof system $\MELLj$ for $\MELL$, enforcing a coarser proof equivalence.
We topologically characterized 
mixed graphs equipped with vertices partitions encoding linear proofs in $\MELLj$,
as well as the graph homomorphisms capturing the resource management part of proofs in $\MELLj$.
Using these results we introduced combinatorial proofs for $\MELL$ and then exponentially handsome proof nets for $\MELL$, defined as compositions of combinatorial proofs.
We also provided a normalization procedure by means of graph rewriting rules.

The \RGB-cographs can be interpreted as a 
coherent interaction graphs~\cite{thom:tit} for $\MLL\cup\set{\krule}$.
We foresee a further application of exponentially handsome proof nets to explore the geometry of interaction of $\MELL$~\cite{GoI} using the same approach from \cite{ehrh},
where handsome proof nets are employed to investigate coherent spaces of $\MLL$.

\paragraph{On Removing Jumps.}

Is it possible to modify exponentially handsome proof nets by removing jumps and thus recover a less coarse notion of proof equivalence including jump rewiring, at the price of loosing the polynomiality of the correctness criterion.

For this purpose, 
we should consider the following set of rules:
$$
\set{\axrule,\onerule,\ljump, \lpar, \ltens,\mixr, \krule,\derrule,\digrule,\jdigrule, \jbotrule, \jewrule,\ecrule}
\qquad\mbox{ with }\qquad
\vlsmash{\vliinf{}{\mixr}{\Gamma, \Delta}{\Gamma}{\Delta}}
\qquad\mbox{ and }\qquad
\vlsmash{\vlinf{}{\ljump}{\ljump}{}}
$$
A correctness criterion for \RGB-cographs encoding linear proofs of this proof system, 
that is, derivations containing only rules in $\set{\axrule,\onerule,\ljump, \lpar, \ltens, \mixr,\krule}$, 
is obtained by dropping the \ae-connectness from \ref{def:Xcorrect}.

Note that we cannot get rid of $\ljump$ in presence of the restricted weakening rule of linear logic: without placeholders we could not represent the proof below on the left, because the skew lifting condition (see \Cref{eq:skewlift}) would fail for any $\uedge$-edge connecting a vertex in $A$ with a vertex in $B$.
\vspace{2pt}
$$
\begin{array}{c@{\qquad\qquad}|@{\qquad\qquad}c}
\vlsmash{\vlderivation{
	\vliin{}{\ltens}{\Gamma, \Delta, A\ltens \wn B}{\vlpr{\pi}{}{\Gamma, A}}{\vlin{}{\ewrule}{\wn B,\Delta}{\vlpr{\pi'}{}{\Delta}}}
}}
&
\vlsmash{\vlderivation{
\vliin{}{\land}{\Gamma, \Delta, A\land B}{\vlpr{\pi}{}{\Gamma, A}}{\vlin{}{\Wrule}{\Delta, B}{\vlpr{\pi'}{}{\Delta}}}
}}
\qquad\rightsquigarrow\qquad
\vlderivation{
{\vliq{}{\Wrule}{\Gamma, A\land B, \Delta}{\vlpr{\pi'}{}{\Delta}}}
}
\end{array}
$$
\vspace{2pt}
The use of \emph{excising}, that is, the transformation above on the right cannot be used to overcome this problem as it is done in classical logic (see \cite{hughes:pws,acc:str:18}) since, in linear logic, this proof transformation cannot be performed.

\paragraph{On Generalized $\wn$-Nodes.}

The rule permutations allowing to push weakening and contraction below the promotion rule (see last line of \Cref{fig:proofEqNS}) are allowed in $\MELL$ proof nets, thanks to the so-called \emph{generalized $\wn$-nodes}~\cite{dan:reg:hilbert}.
The syntax with generalized $\wn$-nodes captures these equivalences, allowing a ``flexible'' representation of boxes, which are the syntactic elements representing promotion rules.
In fact, generalized $\wn$-nodes allow to  identify nets where boxes may have  different amounts of auxiliary ports.
This depends on the fact that boxes have a different status with respect to the interaction net syntax~\cite{lafont:95}.

In  exponentially handsome proof nets, each (weak) promotion is rigidly encoded: each $\oc$-vertex is the principal port of a box and each auxiliary port is encoded by a $\wn$-vertex in the same $\axlink$-class.
In particular, this allows us to  keep track of the depth of a propositional letter $a$,  a $\ljump$ or  a $\lone$ in terms of the number of incoming $\dedge$ in the \RGB-cograph. 
This clearly prevents any rule permutation which may change the number of auxiliary ports of a box.

\subsection*{Acknowledgements}
I wish to thank the anonymous reviewers, 
who gave very valuable feedback and pointed me to relevant literature.
I also would like to thank Christian Retoré, Lutz Stra\ss burger, Giti Omidvar for the useful feedbacks and technical discussions.

\bibliographystyle{eptcs}
\bibliography{main}

\end{document}